\documentclass[12pt,english]{article}
\usepackage{float}
\usepackage[T1]{fontenc}
\usepackage[utf8]{inputenc}
\usepackage[table]{xcolor}
\usepackage[pdftex]{graphicx}
\usepackage{textgreek}
\usepackage{geometry}
\usepackage{textcomp}
\usepackage{natbib}
\usepackage{mathtools}

\usepackage{amsmath}
\usepackage{amssymb}
\usepackage{enumitem}
\usepackage{bm}
\usepackage{graphicx}
\usepackage{caption}
\usepackage{color}
\usepackage{xcolor}
\usepackage{lipsum}
\usepackage{tablefootnote}
\usepackage{setspace}
\mathtoolsset{showonlyrefs} 
\usepackage{mathtools}

\usepackage{babel}
\usepackage{hyperref}
\usepackage{tikz} 
\usetikzlibrary{calc}
\usetikzlibrary{shapes}
\hypersetup{colorlinks=true,linkcolor = blue, urlcolor  = blue,citecolor = blue,anchorcolor = blue}
\usepackage{pgfplots}
\pgfplotsset{compat=1.18}
\usepackage[normalem]{ulem}

\usepackage{makecell}
\usepackage{tcolorbox}
\usepackage{soul,xcolor}
\usepackage{multirow}
\usepackage{tablefootnote}
\usepackage{nicematrix}

\definecolor{red}{rgb}{1,0,0}
\definecolor{gred}{rgb}{0.88,0,0}
\definecolor{redish}{rgb}{0.8,0,0}
\definecolor{green}{rgb}{0,0.45,0}
\definecolor{llgreen}{rgb}{0.9,1,0.9}
\definecolor{ggreen}{rgb}{0,0.4,0}
\definecolor{blue}{rgb}{0,0,1}
\definecolor{gblue}{rgb}{0,0,0.88}
\definecolor{lblue}{rgb}{0.5,0,1}
\definecolor{grey}{rgb}{0.6,0.6,0.6}

\definecolor{lgrey}{rgb}{0.88,0.88,0.88}
\definecolor{orange}{rgb}{1, 0.6, 0}
\definecolor{gorange}{rgb}{0.8, 0.48, 0}
\definecolor{lorange}{rgb}{1, 1, 0.5}
\definecolor{purple}{RGB}{153, 50, 204}

\newcommand{\I}{\mathcal N}
\newcommand{\N}{\mathcal N}
\newcommand{\T}{\mathcal T}
\newcommand{\E}{\mathbb E}
\newcommand{\dd}{\,\mathrm{d}}
\newcommand{\one}{\mathbf{1}}
\renewcommand{\P}{\mathbb P}

\newcommand{\R}{\mathbb R}
\newcommand{\Rplusz}{\mathbb{R}_{\ge 0}}
\renewcommand{\S}{\mathcal S}

\newcommand{\lambdaBar}{1.283} 
\newcommand{\lambdaCritNum}{1.282807777693} 
\newcommand{\lambdaShare}{1.147} 
\newcommand{\lambdaShareCritNum}{1.1463251877} 

\newcommand{\bigwhere}{\ \big|\ }
\newcommand{\Bigwhere}{\ \Big|\ }

\newcommand{\argmax}{\mathop{\arg\max}}

\newcommand{\eps}{\varepsilon}

\usepackage{amsthm}
\newtheorem{Theorem}{Theorem}[section]
\newtheorem{Note}[Theorem]{Note}
\newtheorem{Def}[Theorem]{Definition}

\newtheorem{Lemma}[Theorem]{Lemma}
\newtheorem{Proposition}[Theorem]{Proposition}
\newtheorem{Question}[Theorem]{Question}
\newtheorem{Example}[Theorem]{Example}

\definecolor{playerblue}{RGB}{70,130,180}
\definecolor{playerred}{RGB}{220,20,60}
\definecolor{playergreen}{RGB}{0,128,0}

\title{Prior-free Collusion-proof Dynamic Mechanisms}
\author{Endre Csóka\thanks{Alfréd Rényi Institute of Mathematics, Budapest, Hungary. Supported by the NRDI grant KKP~138270. \textit{Email}: \href{mailto:csoka.endre@renyi.hu}{csoka.endre@renyi.hu}.}}
\date{}

\begin{document}

\maketitle

\begin{abstract}
For a general class of dynamic stochastic multi-player problems, \cite*{CLRT} proposed prior-dependent efficient collusion-proof mechanisms.
The present paper proves prior-free lifting theorems, at the price of lower guaranteed utility levels that depend on the set of possible initial type profiles.
As a special case, we implement a \(0.872\)-approximately utility-maximizing prior-free collusion-proof mechanism for the Markovian repeated single-good allocation problem studied by \cite*{FBT24}.
\end{abstract}

\section{Introduction}

\cite*{CLRT} introduced collusion-proof dynamic mechanisms based on guaranteed utilities.
They call the transferable version the Guaranteed Utility Mechanism (GUM), and the transfer-free version the non-transferable-utility GUM (NTU-GUM).
The key feature is playerwise: each player has a strategy that secures a prescribed expected utility level even if all other players coordinate against her.
In the transferable-utility case, these secured levels can be chosen so that their sum equals the maximum attainable total welfare.
Thus efficiency is implemented through individual guarantees rather than through dominant-strategy truthfulness.

A critical assumption in that construction is that the initial type profile is fixed and built into the mechanism.
This is restrictive when the players' initial types are themselves private information.
The purpose of this paper is to remove this assumption.
The prior-free mechanism first asks the players to report their initial types and then runs the guaranteed-utility mechanism for the reported anticipated environment, with adjustments that preserve the guarantee of a truthful player.
Thus prior-free means initial-profile-free: the mechanism is not indexed by the true realized initial type profile.
In applications, an initial type may itself be a prior distribution, a Markov-chain parameter, or another descriptor of the player's stochastic environment.

We use earlier guaranteed-utility results only through a black-box interface.
In a fixed anticipated environment, TU-GUM takes an admissible total-welfare-maximizing policy and returns a balanced-transfer mechanism together with playerwise guarantee strategies; the resulting guaranteed utility levels sum to the expected total welfare of that policy and are secured against the allowed collusive strategies of the other players.
The NTU-GUM black box gives the corresponding transfer-free approximation after a positive rescaling of utilities under which the chosen policy is total-welfare maximizing, with losses controlled by the bounded-effect quantities stated below.
Finally, the GUE theory of \cite*{GUE25} supplies the equilibrium interpretation: a playerwise guarantee is an outside option, and therefore every equilibrium utility is bounded below by the corresponding guaranteed level.
The present paper proves the prior-free lifting step and the single-good feasibility bounds; it does not reprove the general GUM, NTU-GUM, or GUE theories.

A useful way to view the construction is the following.
Suppose that to each player \(i\), weight \(\alpha_i\), and initial type \(\theta_i\), we assign a target level \(f(\alpha_i,\theta_i)\).
If these target levels are universally feasible, meaning that for every initial type profile their sum can be covered by attainable total welfare, then they can be implemented prior-free as playerwise guaranteed utilities.
That is, by truthfully reporting her own initial type and then following the guarantee strategy, player \(i\) secures \(f(\alpha_i,\theta_i)\) against arbitrary collusive behavior of the other players.

The guarantee strategy is not meant to be an equilibrium prediction of truthful behavior.
It is an outside option within the mechanism.
Consequently, under a common-prior interpretation, every Bayesian Nash equilibrium gives each player at least her guaranteed level, and hence gives total expected utility at least the prior expectation of the sum of the guaranteed levels.
This is the sense in which approximate guaranteed utilities imply approximate efficiency.
The advantage of the GUE formulation for this paper is robustness: the guarantee strategy secures the stated level even against collusive behavior by the other players, and approximate GUE yields welfare lower bounds for all equilibria of the induced game.

As a quantitative application, we specialize the lifting theorems to repeated single-good allocation.
The Markovian value-process model studied by \cite*{FBT24} is a special case: private values evolve over time, and in each period at most one player receives the good.
The framework used here is more general, since initial types may encode value distributions, Markov parameters, or other descriptors of a player's stochastic environment, and the dynamics may also include private decisions, hidden actions, and contractible private outputs affecting future states or payoffs.

For weights \(\alpha_i\) and value distributions \(D_i\), define the fair-floor vector \(g\) by
\[
    g_i=f^*(\alpha_i,D_i).
\]
The fair-floor vector is feasible: the randomized CDF-rank allocation rule gives player \(i\) expected utility \(g_i\) and assigns the good according to the prescribed shares.
The prior-free lifting theorems then convert these fair-floor target levels, with the domain-dependent reductions required for prior-freeness, into playerwise guaranteed utility levels.

The comparison with \cite*{FBT24} has two natural readings.
A distinctive feature of their mechanism is quota balance: each player receives the good the prescribed number of times.
One may treat this feature as a constraint on admissible mechanisms.
We call the resulting problem the quota-constrained version.
Without this restriction, we call it the quota-free version.
The quota-constrained version has a smaller feasible expected-utility region, and attainable utility levels may be lower in absolute terms.
The relevant comparison is therefore proximity to the quota-constrained feasible frontier.

In the quota-constrained comparison, our allocation rule satisfies the prescribed quotas.
The upper bound is proved for the larger expected-share relaxation,
\[
    \E[w_i(V_\N)]\le \alpha_i
    \qquad(i\in\N),
\]
so it also bounds the quota-constrained problem.
In this relaxed class, no feasible expected-utility vector strictly dominates \(\lambdaShare g\).
Since the fair-floor rule itself satisfies the prescribed quotas, this gives the factor
\[
    1/\lambdaShare \approx 0.872.
\]
This is the \(0.872\)-approximately utility-maximizing statement highlighted in the abstract.

Alternatively, one may regard quota balance not as a constraint, but as a proof device for the quota-free problem.
In any quota-free allocation rule, let \(q_i\) be player \(i\)'s expected share of the good.
Since \(\sum_i q_i\le 1=\sum_i\alpha_i\), at least one player satisfies \(q_i\le\alpha_i\).
The one-player bound behind \cite*{FBT24} then rules out that this player receives more than \(e/(e-1)\) times her fair-floor utility \(g_i\).
Hence no quota-free feasible expected-utility vector can strictly dominate
\[
    \frac{e}{e-1}g .
\]
Thus the \(1-1/e\) bound also yields a quota-free feasible-region comparison, even though \cite*{FBT24} do not formulate it in these terms.

Our direct quota-free analysis improves this comparison.
We prove that no quota-free feasible expected-utility vector strictly dominates \(\lambdaBar g\).
Therefore the fair-floor guarantees have factor
\[
    1/\lambdaBar \approx 0.779
\]
for the quota-free feasible-region comparison.
The quota-constrained factor \(0.872\) is higher because the quota-constrained feasible region is smaller: the mechanism may achieve lower utilities in absolute terms, but it is closer to the relevant constrained optimum.

This paper is part of a series leading to the results presented in \cite{CsE}, where initial types are reported via a tendering mechanism.
That mechanism remains applicable under relaxations of risk neutrality and quasi-linear utility, at the cost of a proportional loss in efficiency.
We expect the same technique to carry over to the prior-free mechanisms developed here; we will elaborate on this in follow-up papers.

\section{Illustrative example: allocation problem}
\label{sec:illustrative-example}

This section illustrates the prior-free transformation in a concrete allocation problem; the general theorem is stated in Section~\ref{sec:prior-free-tu}.

We consider the problem of allocating a single indivisible good to one of the players $\N = \{1,2,\dots,n\}$.
If player $i$ receives the good, her value is $V_i$.
Given a prior $V_{\mathcal N}\sim \prod\limits_{i\in\mathcal N} D_i$, TU-GUM implements the efficient allocation.
We consider the model with the following timing:
\begin{enumerate}
    \item Each player $i$ learns the prior $D_i$.
    \item Each player $i$ reports a prior $\widehat D_i$.
    \item Each player $i$ privately observes her value $V_i$, where $V_{\mathcal N}\sim \prod\limits_{i \in \mathcal N} D_i$.
    \item Each player $i$ reports a value $\widehat V_i$.
    \item The good is allocated to $w\in \argmax\limits_i \widehat V_i$ (ties are broken uniformly at random), and we apply a transfer rule $y_{\mathcal N}(\widehat D,\widehat V)$ with $\sum\limits_{i\in\mathcal N} y_i(\widehat D,\widehat V)=0$.
    Every player $i$ aims to maximize $\E(U_i) = \E\big(w(i) \cdot V_i + y_i\big)$ where $w(i) = 1$ if $w = i$, otherwise $w(i) = 0$. 
\end{enumerate}

To every prior $D_i$ we associate a \emph{fair floor} utility level $f(D_i)$: this is the expected utility a player would obtain under efficiency and truthfulness if all players’ priors were $D_i$.
For example, 
\[
f\big({\rm Uniform}[a,b]\big)=\frac{a+n\,b}{n(n+1)},
\]
because if one draws $n$ independent uniform values on $[a,b]$, the expected maximum is $\frac{a+n\,b}{n+1}$.

The prior-free TU-GUM applies the TU-GUM from \cite*{CLRT} as a black box. The reported priors $\widehat D_i$ will be the fixed initial types in TU-GUM, together with constant additional transfers that guarantee each player $i$ the level $f(D_i)$ provided $i$ follows a truthful strategy. As shown in the proof of Theorem~\ref{thm:prior_free_tu_impl}, such constants always exist (typically with a significant freedom).

Consider the following concrete example:
\begin{alignat}{5}
    D_1 &= {\rm Uniform}[2, 14] \hspace{10mm} & D_2 &= {\rm Uniform}[5, 11] \hspace{10mm} & D_3 &\equiv 8 \\
    f(D_1) &= 3\frac{2}{3} & f(D_2) &= 3\frac{1}{6} & f(D_3) &= 2\frac{2}{3}
\end{alignat}
A short calculation yields that the expected total surplus under the efficient allocation satisfies
\begin{equation}
    \sum_{i\in\mathcal N} \E\big(U_i\big)
    \;=\; \E\big(\max_{i\in\mathcal N} V_i\big)
    \;=\; 9\frac{15}{16}
    \;=\; \sum_{i\in\mathcal N} f(D_i)\;+\;\frac{7}{16}.
\end{equation}
The fact that the efficient total surplus can cover these fair-floor guarantees is not specific to this example (see Theorem~\ref{thm:TU_Markov}).
We may distribute this surplus arbitrarily across players via nonnegative constants $c_1(\widehat D),c_2(\widehat D),c_3(\widehat D) \ge 0$ with $c_1(\widehat D)+c_2(\widehat D)+c_3(\widehat D)=\frac{7}{16}$.
The balanced constant transfer used below is
\[
\tau(c)=\left(c_1-\frac{61}{48},\ c_2+\frac{1}{6},\ c_3+\frac{2}{3}\right),
\qquad \sum_{i=1}^3 \tau_i(c)=0.
\]
If the reported priors coincide with the above, i.e., $\widehat D_i=D_i$ for all $i\in\mathcal N$, then the transfer rule is
\begin{equation}
\label{eq:y-rule}
\begin{aligned}
    y_{\mathcal N}(\widehat D_{\N},\widehat V_{\N})
    \;=\;
    &\ \tau\big(c(\widehat D_{\mathcal N})\big)
    +
    \begin{cases}
    \bigl(\ \tfrac{15}{4},\ \tfrac{9}{4},\ -6\ \bigr) 
        & \text{if } \widehat V_1 \le 8 \ \text{and}\ \widehat V_2 \le 8,\\[6pt]
    \bigl(\ \tfrac{15}{4},\ -\tfrac{23}{4},\ 2\ \bigr) 
        & \text{if } \widehat V_1 \le 8 \ \text{and}\ \widehat V_2 > 8,\\[6pt]
    \bigl(\ \tfrac{\widehat V_1^2 - 22\widehat V_1 + 61}{12},\ \tfrac{ - \widehat V_1^2 + 22\widehat V_1 - 85}{12},\ 2\ \bigr)
        & \text{if } \widehat V_1\in(8,11),\ \widehat V_1 \ge \widehat V_2,\\[8pt]
    \bigl(\ \tfrac{\widehat V_1^2-10\widehat V_1 + 61}{12},\ \tfrac{-\widehat V_1^2+10\widehat V_1-85}{12},\ 2\ \bigr)
        & \text{if } \widehat V_1\in(8,11),\ \widehat V_1 < \widehat V_2,\\[8pt]
    \bigl(\ -5,\ 3,\ 2\ \bigr)
        & \text{if } \widehat V_1 \ge 11.
    \end{cases}
\end{aligned}
\end{equation}
With this payment rule, for every $i\in\mathcal N$, if $\widehat D_i=D_i$ and $\widehat V_i=V_i$, the following holds.
\begin{align}
    \mathbb E_{V_i\sim D_i}\!\left[\,U_i\big(\widehat D_{\mathcal N},(V_i,\widehat V_{-i})\big)\,\right]
    &= \mathbb E_{V_i\sim D_i}\!\left[\,\mathbb P\!\big(w(V_i,\widehat V_{-i})=i\big)\cdot V_i
    \;+\; y_i\big(\widehat D_{\mathcal N},(V_i,\widehat V_{-i})\big)\,\right] \\
    &= f(D_i) + c_i\big(D_i,\widehat D_{-i}\big). \label{eq:EU=f+c}
\end{align}
Importantly, player $i$’s utility guarantee is robust even if the other players collusively aim to minimize $i$’s utility.

More generally, for any nonnegative vector-valued function $c=(c_i)_{i\in\N}$ satisfying
\begin{equation}\label{eq:surplus-split-general}
\sum_{i\in\N}\bigl(f(\widehat D_i)+c_i(\widehat D_{\N})\bigr)
=
\E_{V_{\N}\sim\prod_{i\in\N}\widehat D_i}\!\left[\max_{i\in\N}V_i\right],
\end{equation}
\eqref{eq:EU=f+c} holds for every profile $\widehat D_{\mathcal N}$, implying the following.
\begin{equation}
\forall\, \widehat V_{-i}: \hskip 10mm \mathbb E_{V_i\sim D_i}\Big(U_i\big(V_i, (D_i,\widehat{D}_{-i}), (V_i,\widehat{V}_{-i})\big)\Big) \ge\ f(D_i) \hskip 20mm
\end{equation}

It is important to note that none of this implies that we expect truthful reporting. What it does imply is that players will choose strategies under which their expected utilities are at least their guaranteed utility levels. It then follows that, in every BNE under the common prior assumption, the prior expected total surplus is at least the prior expectation of the sum of the guaranteed utilities. A direction for further research is how to allocate the surplus above the sum of the utility guarantees via the rule $c$, and whether this can be used to implement higher utilities, albeit in a sense weaker than the guaranteed utility notion.
This note is equally valid for the following transfer-free mechanism.

\subsection{Transfer-free mechanism for a repeated game}

We now study the model in which the game is repeated for $T$ rounds and transfers are not allowed: $y \equiv 0$.
As a further simplification, we assume that each $D_i$ is the same in every round, so the prior $\widehat D_i$ needs to be reported only at the beginning (thus only steps 3--5 repeat in $T$ rounds).
In this case, for the same three prior distributions $D_{\N}$ as in the TU setup, the guaranteed utilities are as follows.

\begin{alignat}{3}
    D_1 &= {\rm U}[2, 14] \hspace{10mm} & f^-_T(D_1) &= 3\frac{2}{3}\, T \;-\; 10.26 \,\sqrt{T}\,\ln T \;-\; 13.25 \\
    D_2 &= {\rm U}[5, 11] \hspace{10mm} & f^-_T(D_2) &= 3\frac{1}{6}\, T \;-\; 8.06 \,\sqrt{T}\,\ln T \;-\; 10.41 \\
    D_3 &\equiv 8 \hspace{10mm} & f^-_T(D_3) &= 2\frac{2}{3}\, T \;-\; 5.86 \,\sqrt{T}\,\ln T  
\end{alignat}

Thus, in exchange for having no transfers, the utility guarantees incur a sublinear loss in $T$.
We note that this error term is not optimal; we conjecture it can be reduced to order $\ln(T)$ with more advanced techniques; see Question~\ref{que:gue-uniform}.

\bigskip
\noindent
\textbf{Mechanism.} Assume that the reported prior distributions coincide with the above, i.e., $\widehat{D}_{\N} = D_{\N}$. Then, in round $t$, the good is allocated to
\begin{equation}
w_t \;=\; \argmax\big(\,0.84\,\widehat V_{t,1},\; 0.96\,\widehat V_{t,2},\; \widehat V_{t,3}\,\big)
\end{equation}
as long as $\sum\limits_{t=1}^{t_0} \gamma^{i \to j}_{t} > -\,B_{i \to j}$; if this inequality fails for any two players $i, j \in \mathcal N$, $i \ne j$, after some round $t_0$, then player $i$ is excluded from that and all subsequent rounds.
These virtual payments $\gamma$ and budget constraints $B$ are as follows.

\setlength{\jot}{4pt}
\begin{align*}
\gamma^{1 \to 2}_{t} &=
\begin{cases}
  \frac{5266}{5103} & \text{if } 0.84\,\widehat V_{t,1} \le 8,\\[2pt]
  \Bigl(\frac{121}{12} - \frac{16658}{5103}\Bigr)\;-\;\Bigl(\frac{(0.84/0.96)^2}{12}\Bigr)\,\widehat V_{t,1}^{2}
    & \text{if } 8<0.84\,\widehat V_{t,1}<10.56,\\[2pt]
  -\,\frac{16658}{5103} & \text{if } 0.84\,\widehat V_{t,1} \ge 10.56,
\end{cases}
\\[4pt]
\gamma^{1 \to 3}_{t} &=
\begin{cases}
  \frac{940}{567} & \text{if } 0.84\,\widehat V_{t,1} \le 8,\\
  -\,\frac{1580}{567} & \text{if } 0.84\,\widehat V_{t,1} > 8,
\end{cases}
\\[4pt]
\gamma^{2 \to 1}_{t} &=
\begin{cases}
  \displaystyle \frac{\,11\,\widehat V_{t,1} - (0.84/0.96)\,\widehat V_{t,1}^2\,}{6}
    & \text{if } 8 \le 0.84\,\widehat V_{t,1} < 10.56 \ \text{and}\ 0.96\,\widehat V_{t,2}\le 0.84\,\widehat V_{t,1},\\[6pt]
  \displaystyle \frac{\,5\,\widehat V_{t,1} - (0.84/0.96)\,\widehat V_{t,1}^2\,}{6}
    & \text{if } 8 \le 0.84\,\widehat V_{t,1} < 10.56 \ \text{and}\ 0.96\,\widehat V_{t,2} > 0.84\,\widehat V_{t,1},\\[2pt]
  0 & \text{otherwise},
\end{cases}
\\[4pt]
\gamma^{2 \to 3}_{t} &=
\begin{cases}
  \frac{32}{9} & \text{if } 0.84\,\widehat V_{t,1} \le 8 \ \text{and}\ 0.96\,\widehat V_{t,2} \le 8,\\
  -\,\frac{40}{9} & \text{if } 0.84\,\widehat V_{t,1} \le 8 \ \text{and}\ 0.96\,\widehat V_{t,2} > 8,\\
  0 & \text{if } 0.84\,\widehat V_{t,1} > 8,
\end{cases}
\\[6pt]
B_{1 \to 2} &= 4.03\,\sqrt{T}\,\ln T, \qquad
B_{1 \to 3} = 2.93\,\sqrt{T}\,\ln T,\\
B_{2 \to 1} &= 5.13\,\sqrt{T}\,\ln T, \qquad
B_{2 \to 3} = 2.93\,\sqrt{T}\,\ln T.
\end{align*}

In the present case, the anticipated per-period utility profile is asymptotically approximately $(3.79,\,3.26,\,2.79)$, meaning that given these reports of the priors $\widehat{D}_{\N}$, each player $i$ with true prior $D_i=\widehat{D}_i$ is guaranteed to get the respective utility level even if the other two players collusively aim to minimize the utility of $i$.
Recall that the guaranteed utilities without this assumption on $\widehat{D}_{\N}$ are approximately $(3.67,\,3.17,\,2.67)$, in the sense that each player can secure the respective utility level by reporting truthfully $\widehat{D}_i = D_i$ and $\widehat{V}_{t, i} = V_{t, i}$, even if all other players collusively aim to minimize the utility of $i$.

It is important to note that how this surplus is split depends on the constants $0.84$ and $0.96$, for which there is flexibility analogous to the TU case.
This flexibility is, however, universally limited: Theorem~\ref{thm:1.283} shows that no feasible expected-utility vector can strictly dominate \(\lambdaBar\,g\) coordinatewise, where \(g_i=f^*(\alpha_i,D_i)\) denotes the fair-floor target levels.

\subsection{Dependent rounds, hidden actions, public random events}

The same idea works much more generally.
For example, in a Markovian model with a positive decorrelation parameter as defined in~\cite*{BFT23}, the same construction yields the same asymptotic guarantee, but with a larger sublinear error term, proportional to the maximum possible anticipated externality of a single report.
More broadly, in addition to private random events, we may also allow common public random events; the players may take hidden actions; and they may be affected by public decisions or by contractible outputs of other players.
These extensions encompass, for instance, costly investments that improve forecasts of future preferences, or improve the ability to adjust needs to a schedule.

\section{The general setup}

There is a set $\N=\{1,2,\ldots,N\}$ of players, index 0 will refer to the public state, $\N_0 = \N \cup \{0\}$.
Time is discrete, and the number of periods is finite.
The time periods are indexed by $\T = \{0,1,\ldots,T\}$.
In each period $t \in \T$, all players observe the (contractible) public state $\theta_{0,t}\in\boldsymbol{\Theta_0}$, and each player $i\in\N$ privately observes his type $\theta_{i,t}\in\boldsymbol{\Theta_i}$.
We denote the state space by $\boldsymbol{\Theta} = \prod_{i=0}^{N}\boldsymbol{\Theta_i}$.
After the state $\theta_{\N_0,t} \in \boldsymbol{\Theta}$ is realized and any reports required by the mechanism are made, the period choices $x_{0,t}\in\boldsymbol{X}_0$ and $x_{i,t}\in\boldsymbol{X}_i$ are selected, and in the TU case transfers $y_{i,t}\in\mathbb{R}$ are assigned with $\sum\limits_{i\in \N}y_{i,t}=0$.

The model has different versions depending on whether the initial type profile is fixed, drawn from a common prior, or treated as a nondeterministic parameter.
In the prior-free implementation statements below we use $\Xi$ for the private initial-type domain appearing in target functions; in applications an element of $\Xi$ may itself be a distribution, a Markov-chain parameter, or another model descriptor.
When players have different initial-type domains, the notation is read componentwise.

Given any state $\theta_{\N_0,t} \in \boldsymbol{\Theta}$ and decisions $x_{\N_0, t} \in \boldsymbol{X} = \prod\limits_{i \in \N_0} \boldsymbol{X}_i$ at period $t$, the ($t+1$)-period state $\theta_{\N_0, t+1}$ is a random variable distributed according to the probability measure 
\begin{equation}
    \mu(\theta_{\N_0,t}, x_{\N_0, t}) = \mu_{0}(\theta_{0, t},x_{0, t})\times\prod_{i\in\N} \mu_{i}(\theta_{0, t}, \theta_{i, t}, x_{0, t}, x_{i, t})
\end{equation}
where $\mu_{0}:\boldsymbol{\Theta_{0}}\times\boldsymbol{X}_0\rightarrow\Delta(\boldsymbol{\Theta_0})$ and $\mu_{i}:\boldsymbol{\Theta_{0}}\times\boldsymbol{\Theta_{i}}\times\boldsymbol{X}_0 \times \boldsymbol{X}_i \rightarrow\Delta(\boldsymbol{\Theta_i})$.

\paragraph{Utilities.}
We consider both the transferable utility (TU) and non-transferable utility (NTU) settings. In the TU setting, $y_{i,t}\in\mathbb{R}$ denotes a transfer to player $i\in\N$ at period $t$.
The utility of each player $i\in\N$ is given by
\begin{equation}
U_i(\theta_{\N_0, \T}, x_{\N_0, \T}, y_i)=\sum_{t=1}^{T}\big(v_i(\theta_{0,t},\theta_{i,t},x_{0,t},x_{i,t})+y_{i,t}\big).
\end{equation}
That is, we have an environment with \textit{private values} since the per-period valuation function $v_i(\theta_{0,t},\theta_{i,t},x_{0,t},x_{i,t})$ of each player directly depends only on the public state and decision, and on his own private type and private decision. Intuitively, this means that each player can calculate their utility by observing the public history and their own private history, even if they are uncertain about the other players' types. Note that although the players' utilities are quasilinear (in the TU setting), we assume they can depend on private signals and public and private decisions in a general way.

The NTU setup is the same except that $y \equiv 0$. Note that in this case, multiplying the utility function of any player by a positive constant provides an equivalent representation of the setup.

Note that the public and private decisions may become part of the next-period type.
Therefore, our model can express setups where the valuation function depends on the public and private actions.

\paragraph{Timing with the revelation mechanism.}
At period $t=0$, the public initial state $\theta_{0,0}\in\Theta_0$ is publicly observed, and each player $i\in\N$ privately observes his initial type $\theta_{i,0}\in\Theta_i$.
Each player then publicly reports an initial type $\widehat\theta_{i,0}\in\Theta_i$; the reported vector $\widehat\theta_{\N,0}$ is the initial type profile used by the prior-free mechanism.
The period-$0$ public and private decisions $x_{0,0},x_{\N,0}$, when present, only initialize the subsequent stochastic process; the payoff formulas below sum over periods $1,\ldots,T$.
In each payoff period $t=1,\ldots,T$, the public state $\theta_{0,t}$ is observed and each player $i$ privately observes his current type $\theta_{i,t}$.
The public state is verifiable, so we write $\widehat\theta_{0,t}=\theta_{0,t}$, and each player publicly reports a current private type $\widehat\theta_{i,t}\in\Theta_i$.
Given the public state and the private-type reports, the mechanism computes the public decision $x_{0,t}\in X_0$ and, in the TU version, the transfers $y_{i,t}$; each player also chooses any private decision $x_{i,t}\in X_i$ specified by the application.
The next state is then drawn according to $\mu(\theta_{\N_0,t},x_{\N_0,t})$.
The guaranteed-utility statements below are interpreted with player $i$ truthfully reporting $\widehat\theta_{i,0}=\theta_{i,0}$ and $\widehat\theta_{i,t}=\theta_{i,t}$ for $t\ge1$, while the other players may use collusive strategies.

\paragraph{Strategies.}
Recall that the general version of $\eps$-Guaranteed Utility Equilibrium ($\eps$-GUE, \cite*{GUE25}) uses
basic and collusive strategies.
The smaller the sets of basic strategies $\S_i$ and the larger the sets of collusive strategies $\bar{\S}_{-i}$ and $\bar{\S}$ the stronger the meaning of the utility guarantees.
The set of strategies is the set of all probability distributions on mappings from the possible information to the action sets.
The actions of a player are the initial report, the current-type reports $\widehat{\theta}_{i,t}$, and the private decisions $x_{i,t}$.
So we only need to specify the basic information and the collusive information.

The basic information of player $i$ includes his own initial type, his current and past private types, his private decisions, the public states and decisions, and the public reports observed before the decision or report in question.

The collusive information of all players is the entire history of the game until the decision in question.
For the coalition $\N\setminus\{i\}$, the same-round private type $\theta_{i,t}$ of player $i$ is not observed when the period-$t$ reports $\widehat\theta_{\N\setminus\{i\},t}$ are chosen, for $t\ge1$.
Past private types may be included in the collusive information.
For the initial report at $t=0$, the prior-free guarantees below are pointwise in the reported vector $\widehat\theta_{\N,0}$; hence they remain valid whether or not the coalition knows $\theta_{i,0}$ before choosing its own initial reports.

\paragraph{Exiting and re-entering.}
The Guaranteed Utility Mechanism in \cite*{CLRT} allows each player to exit and re-enter the game. Hence, it handles individual rationality in a strong sense, including collusive strategies in which some players leave the game in order to help others.
All these extensions are compatible with the prior-free extension in this paper.

\section{Prior-free transferable utility mechanisms}
\label{sec:prior-free-tu}

\begin{Def}
    We call $f: [0, 1] \times \Xi \to \R$ a {\bf target} function.
     It is {\bf universally feasible} if $\forall \theta_{\N, 0} \in \Xi^{\N}$ and $\forall\big(0 \le \alpha_1, \alpha_2, ..., \alpha_n \bigwhere \sum\limits_{i \in \N} \alpha_i = 1\big)$ there exists a decision policy $\chi$ such that $\E\big(\sum\limits_{i \in \N} U_i(\chi) \bigwhere \theta_{\N,0} \big) \ge \sum\limits_{i \in \N} f(\alpha_i, \theta_{i, 0})$.
\end{Def}

\begin{Def}
    A target function $f$ is {\bf implementable in guaranteed utilities} if, for every weight vector $\alpha_{\N}$ with $\alpha_i\ge0$ and $\sum_{i\in\N}\alpha_i=1$, there exists a revelation mechanism $M(\alpha_{\N})$ such that every player $i$ can guarantee the level $f(\alpha_i,\theta_i)$ for each possible initial type $\theta_i$.
    The guarantee strategy may depend on the publicly fixed weight vector and on player $i$'s own initial type; it includes the initial report $\widehat\theta_{i,0}$, the subsequent current-type reports $\widehat\theta_{i,t}$, and the private decisions.
    Formally, for every such $\alpha_{\N}$, every $i\in\N$, and every $\theta_i\in\Xi$, there exists a strategy $s_i^*(\alpha_{\N},\theta_i)\in\S_i$ such that for every $\theta_{-i}\in\Xi^{\N\setminus\{i\}}$ and every collusive strategy $\bar s_{-i}\in\bar\S_{-i}$,
    \begin{equation}\label{eq:implementable_guarantee_def}
        \E\Bigl(
        U_i\bigl(M(\alpha_{\N}), (s_i^*(\alpha_{\N},\theta_i),\bar s_{-i})\bigr)
        \Bigm| \theta_{\N,0}=(\theta_i,\theta_{-i})
        \Bigr)
        \ge f(\alpha_i,\theta_i).
    \end{equation}
    Thus the mechanism is not indexed by the true initial type vector; it only receives initial types through the players' reports.
\end{Def}

This is a playerwise guarantee notion.
The mechanism does not promise a joint utility profile as an outcome.
Rather, for each player \(i\) and each own initial type \(\theta_i\), it offers a strategy that secures the coordinate \(f(\alpha_i,\theta_i)\), independently of the other players' initial types and of any collusive strategy they use.

We use the prior-dependent guaranteed-utility theorem of \cite*{CLRT} as a black box.

\begin{Theorem}[Prior-dependent TU-GUM; \cite*{CLRT}]\label{thm:imported-tu-gum}
    Fix an initial type profile and a total-welfare maximizing decision policy $\chi$ for that anticipated environment.
    The transferable Guaranteed Utility Mechanism of \cite*{CLRT}, run for this anticipated environment, provides playerwise guarantee levels whose sum equals the expected total welfare generated by $\chi$.
    Each such level is secured by the corresponding player's GUM guarantee strategy against the allowed collusive strategies of the other players.
\end{Theorem}

The following theorem is the main structural result of the paper.

\begin{Theorem}[Prior-free transferable-utility lifting]\label{thm:prior_free_tu_impl}
    If $f$ is universally feasible, then $f$ is implementable in guaranteed utilities.
\end{Theorem}

\begin{proof}
    Fix the weights $\alpha_{\N}$.
    The mechanism first asks every player $i$ to report an initial type $\widehat\theta_{i,0}$.
    For each reported vector $\widehat\theta_{\N,0}$, the mechanism runs GUM for the anticipated initial type profile $\widehat\theta_{\N,0}$.
    Let $C_i(\widehat\theta_{\N,0})$ denote the GUM guaranteed utility of player $i$ in this anticipated instance, using a total-welfare maximizing policy.
    Then $\sum_{i\in\N}C_i(\widehat\theta_{\N,0})$ equals the maximum attainable total utility in the anticipated game.
    By universal feasibility of $f$,
    \[
        \sum_{i\in\N}C_i(\widehat\theta_{\N,0})
        \ge
        \sum_{i\in\N}f(\alpha_i,\widehat\theta_{i,0}).
    \]
    Hence there is a constant balanced transfer vector $\tau(\alpha_{\N},\widehat\theta_{\N,0})$, with $\sum_{i\in\N}\tau_i(\alpha_{\N},\widehat\theta_{\N,0})=0$, such that
    \[
        C_i(\widehat\theta_{\N,0})+\tau_i(\alpha_{\N},\widehat\theta_{\N,0})
        \ge
        f(\alpha_i,\widehat\theta_{i,0})
        \qquad (i\in\N).
    \]
    The prior-free mechanism is GUM for the reported initial profile, plus these constants.

    Now fix a player $i$ with true initial type $\theta_{i,0}$.
    Let $s_i^*(\alpha_{\N},\theta_{i,0})$ be the strategy that reports $\widehat\theta_{i,0}=\theta_{i,0}$ initially and then uses the truthful GUM guarantee strategy for his realized current types.
    The other players may report arbitrary initial types and may use any collusive continuation strategy.
    Conditional on their realized initial reports, the anticipated profile used by the mechanism is $\widehat\theta_{\N,0}=(\theta_{i,0},\widehat\theta_{-i,0})$.
    The GUM guarantee in this anticipated instance applies against arbitrary behavior of the other players; false initial reports by the other players only choose the anticipated instance to which the mechanism is applied.
    Since player $i$'s own initial and current reports are truthful, his guaranteed utility in the true game is therefore at least
    \[
        C_i(\widehat\theta_{\N,0})+\tau_i(\alpha_{\N},\widehat\theta_{\N,0})
        \ge f(\alpha_i,\theta_{i,0}).
    \]
    This proves \eqref{eq:implementable_guarantee_def}.
\end{proof}

\subsection{Sequential good allocation game}

We present a well-studied example in repeated resource allocation~\cite{FBT24, BFT23, GBI21, GBI17}.
The private states evolve via a Markov chain, there are no private decisions, and the public decision can be interpreted as allocating a good in the given round to (at most) one player. 
That is, $x_{0,t} = i$ means that player $i$ obtains a payoff $V(\theta_{i,t})$ at period $t$, and the rest of the players obtain a payoff 0 at that period. 
(The option $x_{0,t} = 0$ means that the good is wasted at period $t$, and every player obtains a payoff 0.)

\begin{Example}\label{exa:TU_Markov}
Let $X_0 = \N \cup \{0\}$, $|X_i| = 1$, $|\Theta_0| = 1$, $\mu: \Xi \to \Delta(\Xi)$, and
\[
    v_i(\theta_{0,t},\theta_{i,t},x_{0,t},x_{i,t}) = {\rm I}(x_{0,t} = i) \cdot V(\theta_{i,t})
\]
for a payoff function $V: \Xi \to \R^+$. 
Let $D_t(\theta_{i,0})$ denote the marginal distribution of the random variable $V(\theta_{i,t})$, conditional on $\theta_{i,0}$.
Assume that $D_t(\theta)$ has finite expectation for every $t\ge1$ and $\theta\in\Xi$. 
\end{Example}

For a distribution $D$ on $\R^+$ with finite expectation, write ${\rm F}_D$ for its cumulative distribution function and
\[
    Q_D(u):=\inf\{x:{\rm F}_D(x)\ge u\}
\]
for its generalized inverse, or quantile function.  We use the following atom-compatible rank construction.  If $V\sim D$ and $Z\sim {\rm Unif}[0,1]$ is independent of $V$, define
\begin{equation}\label{eq:randomized-cdf-rank}
    R_D(V,Z):={\rm F}_D(V-)+Z\bigl({\rm F}_D(V)-{\rm F}_D(V-)\bigr).
\end{equation}
This is the randomized CDF rank, also known as the distributional transform or randomized probability integral transform.  Conditional on $V=v$, the rank is drawn uniformly from the jump interval $[{\rm F}_D(v-),{\rm F}_D(v)]$.  Hence $R_D(V,Z)\sim {\rm Unif}[0,1]$, and
\[
    Q_D(R_D(V,Z))=V\qquad\text{almost surely.}
\]
For $\alpha\in [0,1]$, set $\varphi(0,D)=0$ and, for $\alpha>0$,
\begin{equation}
\varphi(\alpha, D) = \E\bigl[ V \cdot R_D(V,Z)^{1/\alpha - 1} \bigr],
\end{equation}
where $V\sim D$ and $Z\sim {\rm Unif}[0,1]$ is independent.  Equivalently,
\[
    \varphi(\alpha,D)=\int_0^1 Q_D(u)u^{1/\alpha-1}\,du.
\]

Given the setup in Example~\ref{exa:TU_Markov}, let 
\begin{equation}
    f^*(\alpha, \theta)= \sum\limits_{t=1}^{T} \varphi(\alpha, D_t(\theta)).
\end{equation}
We construct a decision policy $\chi^*$ that witnesses the universal feasibility of the above target function $f^*$. 
At period $t$, the decision policy $\chi^*$ may depend on the full history up to that period. 
However, in this example, it only depends on the initial- and current type vectors and on independent atom-breaking randomization.  For a given weight vector, write $\I_+:=\{j\in\I:\alpha_j>0\}$. Since the weights sum to one, $\I_+$ is nonempty.  At each period $t$, draw variables $Z_{j,t}\sim {\rm Unif}[0,1]$, independently across players and periods and independently of the type process, and set
\[
    R_{j,t}:=R_{D_t(\theta_{j,0})}\bigl(V(\theta_{j,t}),Z_{j,t}\bigr)
    \qquad(j\in\I_+).
\]
The policy allocates the good to
\begin{equation}\label{eq:chistar}
\chi^*(\theta_{\I,0}, \theta_{\I,t}, Z_{\I_+,t})=\argmax\limits_{j\in \I_+} R_{j,t}^{1/\alpha_j}.
\end{equation}
Zero-weight players receive no positive guarantee because $\varphi(0,D)=0$.
Via this stochastic decision policy, the good is allocated to exactly one positive-weight player at each period almost surely.

\begin{Lemma}\label{lem:uniforms}
Let $R_j$, $j\in\I_+$, be independent uniform random variables on $[0,1]$, and let $\alpha_j>0$ for $j\in\I_+$ with $\sum_{j\in\I_+}\alpha_j=1$. Then:
\begin{enumerate}
    \item if $0<\alpha_i<1$, then for any $i\in\I_+$ we have
    \[
        \max_{j\in\I_+\setminus\{i\}} R_j^{1/\alpha_j}
        \sim {\rm Uniform}(0,1)^{1/(1-\alpha_i)};
    \]
    \item for any $i\in\I_+$ and $x\in[0,1]$,
    \[
        \P\bigl(\argmax_{j\in\I_+} R_j^{1/\alpha_j}=i\mid R_i=x\bigr)
        = x^{1/\alpha_i-1}.
    \]
\end{enumerate}
\end{Lemma}
\begin{proof}
For item 1 and $x\in[0,1]$,
\[
\P\left(\forall j\in\I_+\setminus\{i\}: R_j^{1/\alpha_j}\le x\right)
=\prod_{j\in\I_+\setminus\{i\}}\P(R_j\le x^{\alpha_j})
=x^{\sum_{j\in\I_+\setminus\{i\}}\alpha_j}
=x^{1-\alpha_i}.
\]
For item 2, conditioning on $R_i=x$ gives
\[
\P\bigl(\argmax_{j\in\I_+} R_j^{1/\alpha_j}=i\mid R_i=x\bigr)
=\P\left(\max_{j\in\I_+\setminus\{i\}}R_j^{1/\alpha_j}\le x^{1/\alpha_i}\right)
=x^{(1-\alpha_i)/\alpha_i}
=x^{1/\alpha_i-1},
\]
with the same formula also covering the case $\alpha_i=1$ by the empty-product convention.
\end{proof}

\begin{Theorem}\label{thm:TU_Markov}
The target function $f^*$ is universally feasible via the decision policy $\chi^*$, and $\E\big(U_i(\chi^*) \bigwhere \theta_{\N,0} \big)= f^*(\alpha_i, \theta_{i,0})$.  
\end{Theorem}
\begin{proof}
Fix a player $i$. If $\alpha_i=0$, then $\chi^*$ never selects $i$ and the required lower bound is $0=f^*(0,\theta_{i,0})$.
Assume henceforth that $\alpha_i>0$.
At period $t$, use the ranks $R_{j,t}$ defined in the construction of $\chi^*$ above.  The variables $(R_{j,t})_{j\in\I_+}$ are independent uniform random variables on $[0,1]$. By Lemma~\ref{lem:uniforms}, conditional on $R_{i,t}=s$, player $i$ is selected with probability $s^{1/\alpha_i-1}$.
Using the quantile identity $V(\theta_{i,t})=Q_{D_t(\theta_{i,0})}(R_{i,t})$ almost surely, the expected payoff of player $i$ at period $t$ is therefore
\[
    \int_0^1 Q_{D_t(\theta_{i,0})}(s)s^{1/\alpha_i-1}\,ds
    =\varphi(\alpha_i,D_t(\theta_{i,0})).
\]
Adding over $t=1,\ldots,T$ gives
\[
    \E\big(U_i(\chi^*) \bigwhere \theta_{\N,0} \big)
    =\sum_{t=1}^T\varphi(\alpha_i,D_t(\theta_{i,0}))
    =f^*(\alpha_i,\theta_{i,0}).
\]
\end{proof}


If there are $n$ players with the same initial type $\theta_{i,0}=\theta$ and weight $\alpha_i=1/n$, the decision policy $\chi^*$ clearly maximizes the total utility, since it selects the player with the highest evaluation at each period. 
Hence, in that case, the target function we defined attains the maximum among all universally feasible target functions.

\begin{Proposition}\label{prop:Pareto}
Consider the special case of the setup in Example~\ref{exa:TU_Markov} where we have a $T$-fold repeated game.
That is, each player $i$ is assigned a distribution $D_i$ on $\R_{\ge 0}$ with finite expectation.
The initial types consist of these distributions: $\theta_{i,0}=D_i$.
The type $\theta_{i,t}$ at period $t\geq 1$ is a random sample from $D_i$ (drawn independently over all periods and players), and the valuation function is the identity function, that is, $V(\theta_{i,t}) = \theta_{i,t}$.

In this setup, the target function $f^*$ is Pareto optimal among all universally feasible target functions.
\end{Proposition}
\begin{proof}
Write $Q_D(u):=\inf\{x:F_D(x)\ge u\}$ for the quantile of $D$.
In the present i.i.d. repeated setup,
\[
    f^*(\alpha,D)=T\int_0^1 Q_D(s)s^{1/\alpha-1}\,\mathrm ds
    \qquad (\alpha>0),
\]
and $f^*(0,D)=0$.

Let $f^{**}$ be a universally feasible target function that Pareto dominates $f^*$.
When $n$ players have equal weights $1/n$ and the same type $\theta$, the total welfare is maximized by allocating the good to the highest valuation in every period, and this gives total expected welfare $n f^*(1/n,\theta)$.
Universal feasibility therefore gives $f^{**}(1/n,\theta)\le f^*(1/n,\theta)$, while Pareto domination gives the reverse inequality.
Hence
\begin{equation}\label{eq:pareto-equal-weights}
    f^{**}(1/n,\theta)=f^*(1/n,\theta)
    \qquad(n\in\mathbb N,\ \theta\in\Xi).
\end{equation}

We next prove equality for rational positive weights.
Fix $1\le k\le n$ and a distribution $D$ on $\R_{\ge0}$ with finite expectation.
Define $D[k]$ by its quantile
\[
    Q_{D[k]}(u):=Q_D(u^k),\qquad 0\le u\le 1.
\]
If $W_1,\ldots,W_k$ are i.i.d. with law $D[k]$, then $\max_{1\le r\le k}W_r\sim D$: indeed, writing $W_r=Q_D(U_r^k)$ with independent uniform $U_r$, we have
\[
    \max_{1\le r\le k}W_r
    =Q_D\!\left(\max_{1\le r\le k}U_r^k\right),
\]
and $\max_r U_r^k$ is uniform on $[0,1]$.

Consider an instance with one player of weight $k/n$ and type $D$, and $n-k$ players of weight $1/n$ and type $D[k]$.
Its maximum total expected welfare is the same as in the symmetric instance with $n$ players of weight $1/n$ and type $D[k]$, because the first player's valuation has the same distribution as the maximum of $k$ independent $D[k]$-valuations.
Thus this maximum total expected welfare is $n f^*(1/n,D[k])$.
By universal feasibility of $f^{**}$ and by \eqref{eq:pareto-equal-weights},
\[
    f^{**}(k/n,D)+(n-k)f^*(1/n,D[k])
    \le n f^*(1/n,D[k]),
\]
so $f^{**}(k/n,D)\le k f^*(1/n,D[k])$.
On the other hand, by the change of variables $s=u^k$,
\[
    k f^*(1/n,D[k])
    =Tk\int_0^1 Q_D(u^k)u^{n-1}\,\mathrm du
    =T\int_0^1 Q_D(s)s^{n/k-1}\,\mathrm ds
    =f^*(k/n,D).
\]
Together with Pareto domination, this proves
\begin{equation}\label{eq:pareto-rational-weights}
    f^{**}(k/n,D)=f^*(k/n,D)
    \qquad(1\le k\le n).
\end{equation}

The endpoint $\alpha=0$ also cannot be improved.
For any $m\in\mathbb N$, take $m$ zero-weight players with type $D$ and one weight-one player with the zero distribution $\mathbf 0$.
Universal feasibility and \eqref{eq:pareto-equal-weights} for the zero distribution give
\[
    m f^{**}(0,D) \le T\,\E\max_{1\le r\le m} V_r,
\]
where the $V_r$ are i.i.d. with law $D$.
Since $D$ is integrable, $m^{-1}\E\max_{r\le m}V_r\to0$; hence $f^{**}(0,D)\le0=f^*(0,D)$, and Pareto domination again gives equality.

It remains to rule out an improvement at an irrational weight $\alpha\in(0,1)$.
Suppose that $f^{**}(\alpha,D)>f^*(\alpha,D)$.
By continuity of $f^*(\cdot,D)$ on $(0,1]$, choose a rational $k/n\in(\alpha,1]$ such that
\[
    f^{**}(\alpha,D)>f^*(k/n,D)=k f^*(1/n,D[k]).
\]
Consider the instance with players of weights and types
\[
    (\alpha,D),\qquad (k/n-\alpha,\mathbf 0),\qquad
    \underbrace{(1/n,D[k]),\ldots,(1/n,D[k])}_{n-k\text{ players}}.
\]
Because $D$ is supported on $\R_{\ge0}$, the maximum of the first two players' valuations has law $D$, and this is also the law of the maximum of $k$ independent $D[k]$-valuations.
Therefore the maximum total expected welfare of this instance is again $n f^*(1/n,D[k])$.
But universal feasibility and Pareto domination would require total guaranteed utility at least
\[
    f^{**}(\alpha,D)+f^{**}(k/n-\alpha,\mathbf0)+(n-k)f^{**}(1/n,D[k])
    > kf^*(1/n,D[k])+(n-k)f^*(1/n,D[k]),
\]
which is $n f^*(1/n,D[k])$, a contradiction.
Thus no universally feasible target function can Pareto dominate $f^*$ strictly anywhere.
\end{proof}

The assertion of Proposition~\ref{prop:Pareto} can be generalized to other type spaces. 
The proof works whenever we can split any initial type into $k$ in the following sense: 
For any $\theta\in \Xi$ there is a $\theta[k]\in \Xi$ such that for every $1\leq t\leq T$ we have the distributional identity $Q_{D_t(\theta)}(U^k) \sim D_t(\theta[k])$, where $U\sim\text{Uniform}(0,1)$ and $D_t(\theta)$ is as in Example~\ref{exa:TU_Markov}. 
Moreover, the statement can be made more specific. 
If, under the above assumption, a universally feasible function $f$ attains the optimum on all $(\alpha, \theta)$ where $\alpha=1/n$ for some $n\in \mathbb{N}$, then $f\leq f^*$. 
That is, we can only set utility guarantees that are larger somewhere than $f^*$ at the expense of giving up some of the guarantees in arguably the most relevant situations, when players have equal weights. 

We close the section with an open problem. 
Consider the modification of the model in Example~\ref{exa:TU_Markov} in which the public decision in each period is a subset of at most $\ell$ players. 
That is, we allocate (at most) $\ell$ goods in every round, each to a different player or to nobody. 
There is an obvious way to generalize the construction of the allocation rule $\chi^*$ by saying that, after drawing the atom-breaking ranks $R_{j,t}:=R_{D_t(\theta_{j,0})}(V(\theta_{j,t}),Z_{j,t})$, we pick the indices corresponding to the $\ell$ largest values of the set $R_{j,t}^{1/\alpha_j}$. 
Let us denote the universally feasible function obtained from this allocation rule by $f^*_\ell(\alpha, \theta)$; in particular, $f^*_1=f^*$. 
Unfortunately, this is typically not optimal. 
For instance, if $\ell=2$, then this allocation rule yields the universally feasible function $f^*_2(\alpha, \theta)=f^*(\min \{1, 2\alpha-\alpha^2\}, \theta)$. 
Clearly, for any fixed $\ell\in \mathbb{N}$, the function $f^*(\min \{1, c\alpha\}, \theta)$ cannot be universally feasible for any $c>\ell$. 
Indeed, choose $\eps>0$ so small that $c(1/\ell-\eps)\geq 1$. 
Take a type $\theta$ such that $D_t(\theta)\sim {\rm Uniform}(0,1)$ for all $1\leq t\leq T$. 
Any player of weight $1/\ell-\eps$ would then have target $f^*(1,\theta)$ and hence would have to receive a good in every period almost surely. 
With $\ell$ such players, all $\ell$ goods are exhausted in every period. 
Now add a player of weight $\alpha_0$ with the same type $\theta$, where $0<\alpha_0\leq \ell\eps$, and a zero-value residual player of weight $\ell\eps-\alpha_0$. 
The total weight is one, but the small player can never receive a good, contradicting the positive target $f^*(\min\{1,c\alpha_0\},\theta)>0$. 

\begin{Question} \label{que:lgoods}
Let $\ell\geq 2$. 
Is $f^*(\min \{1, \ell\alpha\}, \theta)$ universally feasible for the above problem of allocating (at most) $\ell$ goods in each period?
\end{Question}

\section{Prior-free non-transferable utility mechanisms}
\label{sec:prior-free-ntu}

\begin{Def}
    A target function $f: [0, 1] \times \Xi \to \R$ is {\bf NTU-universally feasible} if $\forall \theta_{\N, 0} \in \Xi^{\N}$ and $\forall\big(0 \le \alpha_1, \alpha_2, ..., \alpha_n \bigwhere \sum\limits_{i \in \N} \alpha_i = 1\big)$ there exists a stochastic decision policy $\chi$ such that $\E\big( U_i(\chi) \big) \ge f(\alpha_i, \theta_{i, 0})$ for every $i \in \N$.
\end{Def}

\begin{Def}
    A type space $\Xi$ is scalable if the following holds. For each $\theta \in \Xi$ and $c \in \R_{0+}$, there exists a type $c \cdot \theta \in \Xi$ such that for every $\theta_0 \in \Theta_0$, $x_0 \in X_0$, and $x_i \in X_i$,
    \begin{equation}
        v_i(\theta_0, c \cdot \theta, x_0, x_i) = c \cdot v_i(\theta_0, \theta, x_0, x_i),
    \end{equation}
    and
    \begin{equation}
        \mu_i(\theta_0, c \cdot \theta, x_0, x_i) = c \cdot \mu_i(\theta_0, \theta, x_0, x_i).
    \end{equation}
    Here $c \cdot \mu_i(\theta_0,\theta,x_0,x_i)$ denotes the image (pushforward) of the measure $\mu_i(\theta_0,\theta,x_0,x_i)$ under the map $\eta \mapsto c \cdot \eta$.

    A target function is homogeneous if $f(\alpha, c \cdot \theta) = c \cdot f(\alpha, \theta)$.
\end{Def}

\begin{Lemma}
    NTU-universal feasibility implies universal feasibility.
    For homogeneous target functions, these are equivalent.
\end{Lemma}

\begin{proof}
    NTU-universal feasibility means universal feasibility with 0 transfers but with a stochastic rule. From the support of the stochastic rule, we can choose a deterministic rule with at least the same total valuation.
    
    For the second statement, consider a homogeneous universally feasible target function $f$.
    Assume for contradiction that $f$ is not NTU-universally feasible, and let $(\alpha_{\N}, \theta_{\N})$ be a counterexample.
    For a deterministic decision policy $\chi$ on this instance define
    \begin{equation}
        d(\chi)_i = \E\big(U_i(\chi) \bigwhere \theta_{\N}\big) - f(\alpha_i, \theta_i).
    \end{equation}
    If a convex combination of the $d(\chi)$ vectors were nonnegative, then the corresponding stochastic decision policy would be a feasible solution, contradicting the assumption that $(\alpha_{\N}, \theta_{\N})$ is a counterexample.
    Therefore, by a separating-hyperplane argument, there exists a vector $\beta \in \R_+^{\N}$ such that
    \begin{equation}
        \big\langle d(\chi), \beta \big\rangle < 0
    \end{equation}
    for every deterministic decision policy $\chi$.

    Let $\widetilde{\chi}$ be an $f$-feasible deterministic decision policy for the scaled instance $(\alpha_{\N}, (\beta \cdot \theta)_{\N})$.
    Define a deterministic decision policy $\chi_\beta$ for the original instance by
    \begin{equation}
        \chi_\beta(h) := \widetilde{\chi}(\beta \cdot h),
    \end{equation}
    where $\beta \cdot h$ denotes the history obtained from $h$ by scaling player $i$'s private type by $\beta_i$ at every date.
    By the scalability assumption, the law of the scaled history under $\chi_\beta$ in the original instance coincides with the law of the history under $\widetilde{\chi}$ in the scaled instance. Hence, for every $i \in \N$,
    \begin{equation}
        \E\Big(U_i(\widetilde{\chi}) \Bigwhere (\beta \cdot \theta)_{\N}\Big)
        = \beta_i \cdot \E\Big(U_i(\chi_\beta) \Bigwhere \theta_{\N}\Big).
    \end{equation}
    Therefore,
    \begin{align}
        \big\langle d(\chi_\beta), \beta \big\rangle
        &= \sum_{i \in \N} \beta_i \cdot \E\Big(U_i(\chi_\beta) \Bigwhere \theta_{\N}\Big) - \sum_{i \in \N} \beta_i \cdot f(\alpha_i, \theta_i) \\
        &= \sum_{i \in \N} \E\Big(U_i(\widetilde{\chi}) \Bigwhere (\beta \cdot \theta)_{\N}\Big) - \sum_{i \in \N} f(\alpha_i, \beta_i \cdot \theta_i) \\
        &\ge 0,
    \end{align}
    where the last inequality follows from the universal feasibility of $f$ for the scaled instance and the homogeneity of $f$. This contradicts $\big\langle d(\chi_\beta), \beta \big\rangle < 0$.
\end{proof}

\begin{Note}
    If the decision set $X$ includes an arbitrary transfer rule, formally, $X = X_0 \times (t \in \R^{\N} | \sum t_i = 0)$ and $v_i\big(\theta_i, (x_0, t_{\N})\big) = v^{\prime}_i(\theta_i, x_0) + t_i$, then NTU-universal feasibility is equivalent to universal feasibility.
    Accordingly, the results for the TU model can be deduced as a formal consequence of the results for the NTU model in this special case.
\end{Note}

In this section, we no longer fix the number of rounds $T$, but with a focus on repeated games and its extensions, we compare the settings with different values of $T$.
Accordingly, a target function will be denoted by $f_T = f_T(\alpha, \theta)$, and $f = (f_1, f_2, ...)$.
In particular, for independently repeated games, $f_T(\alpha, \theta) = T \cdot f_1(\alpha, \theta)$ is a natural choice and, in this case, the universal feasibility of $f_T$ is equivalent to the universal feasibility of $f_1$.

We use the phrase bounded-effect setup for environments in which the effect of one player's report on another player's future NTU-GUM accounting can be controlled by type-dependent source and sensitivity multipliers.
The next definition gives the formal condition used by the transfer-free application.

\begin{Def}\label{def:externality-multipliers}
    We call ${\rm Sour}(\theta_{i,0})$ an externality source multiplier and
    ${\rm Sens}(\theta_{i,0})$ an externality sensitivity multiplier.
    We say that $(f, {\rm Sour}, {\rm Sens})$ is universally feasible if for every anticipated initial setup $(\alpha_{\N},\vartheta_{\N,0},T)$ there is a decision policy $\chi$ such that
    \[
        \E\big(U_i(\chi)\bigwhere \vartheta_{\N,0}\big) \ge f_T(\alpha_i,\vartheta_{i,0})
        \qquad (i\in\N),
    \]
    and, for every subsequent history, the absolute value of the GUM externality $\gamma_{i\to j,t}$ by $i$ on $j$ with any single report is bounded by ${\rm Sour}(\vartheta_{i,0})\cdot {\rm Sens}(\vartheta_{j,0})$.
\end{Def}

Let
\begin{equation}
    {\rm range}(\theta_{i,0}) := \sup \mathcal{V}(\theta_{i,0}) - \inf \mathcal{V}(\theta_{i,0}),
\end{equation}
where $\mathcal{V}(\theta_{i,0})$ denotes the set of one-period valuations that player $i$ may have in the given setup starting from $\theta_{i,0}$.
Informally, ${\rm range}(\theta_{i,0})$ will only be used as a crude upper bound on the per-round loss that player $i$ may suffer if, after some bad event, the mechanism ignores the true reports of player $i$ and effectively replaces them by arbitrary reports.
In a concrete setup, whenever a smaller valid upper bound on this per-round loss is available, one may replace ${\rm range}(\theta_{i,0})$ by that sharper quantity.

\begin{Example}
    In a repeated game with independent periods, if $f$ is universally feasible and ${\rm Sour} \equiv 1$, ${\rm Sens} \equiv {\rm range}$, then $(f,{\rm Sour},{\rm Sens})$ is universally feasible.
    
    As another example, if $\theta_{i,0}$ indicates a Markov chain with a decorrelation parameter $\gamma(\theta_{i,0})$ as defined in~\cite{FBT24}, then $f$ with ${\rm Sour}(\theta_{i,0}) = \gamma(\theta_{i,0})^{-2}$ and ${\rm Sens}(\theta_{i,0}) = \gamma(\theta_{i,0})^{-2} \cdot {\rm range}(\theta_{i,0})$ is universally feasible.
\end{Example}

\paragraph{Supporting renormalization for finite horizons.}
Fix an initial setup and let $K\subset\R^{\N}$ be the set of expected utility profiles generated by stochastic decision policies.
If a policy induces $u^*\in K$ and $u^*$ has a strictly positive supporting normal, i.e., there is $\beta\in\R^{\N}_{>0}$ such that
\[
    \sum_{i\in\N}\beta_i u_i^* \ge \sum_{i\in\N}\beta_i u_i
    \qquad\text{for every }u\in K,
\]
then, after multiplying player $i$'s utility by $\beta_i$, that policy is total-valuation-maximizing.
For a finite horizon with finite action and type spaces, $K$ is a compact convex polytope; hence this condition is exactly the requirement that the chosen profile lie on an exposed Pareto face with a strictly positive exposing normal.
Indeed, the displayed inequality is precisely maximization of the sum of the rescaled utilities, and the finite case is the standard supporting-hyperplane characterization of exposed faces of the policy polytope.

For the theorem below, say that $(f,{\rm Sour},{\rm Sens})$ admits \emph{supported witnesses} if, for every horizon $T$, every positive weight vector $\alpha_{\N}$ with $\sum_i\alpha_i=1$, and every anticipated initial setup $\vartheta_{\N,0}$, the policy witnessing Definition~\ref{def:externality-multipliers} can be chosen so that its expected utility profile has a strictly positive supporting normal, while satisfying the same single-report externality bounds.

\begin{Theorem}\label{thm:prior-free-ntu}
    If $(f, {\rm Sour}, {\rm Sens})$ is universally feasible and admits supported witnesses, then for every horizon $T\ge1$ and every positive weight vector $\alpha_{\N}$ with $\sum_i\alpha_i=1$ the target levels
    \begin{equation} \label{eq:NTU-guarantee}
    f_T(\alpha_i, \theta_{i,0}) - \sqrt{T} \cdot \ln(T/\alpha_i) \cdot {\rm Sens}(\theta_{i,0}) - {\rm range}(\theta_{i,0}) \cdot \exp \Big(\frac{{\rm Sour}(\theta_{i,0})^2}{2\alpha_i^2}\Big)
    \end{equation}
    are implementable in guaranteed utilities by a prior-free and transfer-free mechanism.
    Zero-weight players, if present in an application, are omitted before applying this statement.
\end{Theorem}

\begin{proof}
    Omit zero-weight players, so $\alpha_i>0$ for every remaining player.
    The mechanism asks that every player $i$ report an initial type $\widehat\theta_{i,0}$.
    Fix a reported profile $\widehat\theta_{\N,0}$.
    For the anticipated setup $(\alpha_{\N},\widehat\theta_{\N,0},T)$, choose a supported witnessing policy $\chi^*$ and a supporting vector $\beta\in\R^{\N}_{>0}$.
    Apply NTU-GUM~\cite*{CLRT} to the rescaled instance in which player $k$'s utility is multiplied by $\beta_k$.
    By the supporting-renormalization observation above, $\chi^*$ is total-valuation-maximizing in this rescaled instance.
    In rescaled units the target, sensitivity multiplier, and range term of player $k$ are multiplied by $\beta_k$; ${\rm Sour}$ and the weights are unchanged.
    Thus, if player $i$ reports the initial type truthfully, the rescaled GUM guarantee before the NTU budget losses is at least $\beta_i f_T(\alpha_i,\theta_{i,0})$.

    Use the budget constraints
    \begin{equation}\label{eq:ntu-budget-rescaled}
        B_{i,j}=\alpha_i\sqrt{T}\,\ln(T/\alpha_j)\,\beta_j {\rm Sens}(\widehat\theta_{j,0}),
    \end{equation}
    measured in rescaled utility units.
    We estimate losses in rescaled units and divide player $i$'s final guarantee by $\beta_i$.

    With NTU-GUM, each player $i$ gets the same guaranteed utility as with GUM except for two sources of error: the virtual transfers to $i$ that are not paid, and the possible loss if $i$ runs out of one of his budgets and his reports are subsequently ignored.

    First consider the virtual transfers to player $i$ that are not paid in the transfer-free implementation.
    By the NTU-GUM budget construction, their total loss is bounded by $\sum_{j\ne i}B_{j,i}$ in rescaled units.
    Therefore the first error, after division by $\beta_i$, is at most
    \[
        \sum_{j\in\N\setminus\{i\}}\frac{B_{j,i}}{\beta_i}
        =\sum_{j\in\N\setminus\{i\}}\alpha_j\sqrt{T}\,\ln(T/\alpha_i)\,{\rm Sens}(\theta_{i,0})
        \le \sqrt{T}\,\ln(T/\alpha_i)\,{\rm Sens}(\theta_{i,0}),
    \]
    where we used truthful initial reporting, $\widehat\theta_{i,0}=\theta_{i,0}$ and $\sum_{j\ne i}\alpha_j\le1$.

    For the budget of $i$ against $j$, the externalities are computed in the reported/anticipated instance. Thus, if $i$ is truthful, then $(\gamma_{i\to j,t})_t$ is a martingale-difference sequence and, in rescaled units,
    \[
        |\gamma_{i\to j,t}|
        \le C_{i,j}:={\rm Sour}(\theta_{i,0})\,\beta_j {\rm Sens}(\widehat\theta_{j,0}).
    \]
    If $C_{i,j}=0$, then each such externality is zero and the bad event below has probability zero.
    Otherwise the one-sided maximal Azuma--Hoeffding inequality gives
    \begin{equation}\label{eq:Azuma1}
        \P\left(\min_{0\le s\le T}\sum_{t=1}^s\gamma_{i\to j,t}<-B_{i,j}\right)
        \le
        \exp\!\left(\frac{-B_{i,j}^2}{2T\,C_{i,j}^2}\right)
        =\exp\!\left(\frac{-\alpha_i^2\ln^2(T/\alpha_j)}{2{\rm Sour}(\theta_{i,0})^2}\right).
    \end{equation}
    Using $-c\ln^2(T/x)\le \ln(x/T)+1/(4c)$ with $c=\alpha_i^2/(2{\rm Sour}(\theta_{i,0})^2)$ and $x=\alpha_j$, the last expression is at most
    \[
        \frac{\alpha_j}{T}\exp\!\left(\frac{{\rm Sour}(\theta_{i,0})^2}{2\alpha_i^2}\right).
    \]
    The same probability bound is trivial in the case $C_{i,j}=0$.
    If this bad event occurs, player $i$ may lose at most $T\beta_i{\rm range}(\theta_{i,0})$ in rescaled units. Hence, after summing over $j\ne i$ and dividing by $\beta_i$, the second error is at most
    \[
        {\rm range}(\theta_{i,0})\exp\!\left(\frac{{\rm Sour}(\theta_{i,0})^2}{2\alpha_i^2}\right).
    \]
    Combining the two error bounds gives \eqref{eq:NTU-guarantee}.
\end{proof}

\section{A quantitative application: repeated single-good allocation}
\label{sec:constant-bounds}

The preceding sections are model-general.
We now specialize to repeated single-good allocation in order to quantify the strength of one canonical family of target levels.

The analysis in this section is an auxiliary feasibility analysis.
It does not introduce a new prior-free mechanism and it does not modify the lifting theorem.
Instead, we define fair-floor target levels \(g_i=f^*(\alpha_i,D_i)\), prove that the vector \(g\) is feasible in the single-good allocation problem, and then ask how much all coordinates of \(g\) could be simultaneously increased while remaining feasible.
By the prior-free TU lifting theorem, feasibility of these targets is enough to implement them as playerwise guaranteed utility levels with transfers.
Their transfer-free interpretation is for the bounded-effect NTU cases covered by Theorem~\ref{thm:prior-free-ntu}.

We consider two versions of the auxiliary single-round problem.
In the quota-free version there is no explicit frequency or share constraint.
The quota-constrained version is motivated by the quota-balance feature of \cite*{FBT24}: each player receives the good the prescribed number of times.

For the quota-constrained upper bound, we work with the expected-share relaxation
\[
    \E[w_i(V_\N)]\le \alpha_i\qquad(i\in\N).
\]
Every quota-constrained rule is feasible for this relaxation, while the fair-floor rule satisfies the prescribed quotas.
Thus an upper bound for the relaxed class is also an upper bound for the quota-constrained comparison.
The relaxed model contains the exact-share no-withholding formulation: any rule with $\E[w_i(V_\N)]=\alpha_i$ and no withholding is admissible here.
Conversely, for the purpose of upper bounds it has the same supremal scale factor.
Indeed, if a rule withholds mass and has shares $q_i:=\E[w_i(V_\N)]\le\alpha_i$, then the withheld mass may be reassigned to players in proportion to the deficits $\alpha_i-q_i$; all shares become exactly $\alpha_i$ and all utilities weakly increase because values are nonnegative.

\subsection{Auxiliary single-round allocation problem and fair-floor target}\label{sec:one-shot-allocation}

Consider a single-round resource allocation problem. There is a set of agents $\N=\{1,\dots,n\}$ and a single indivisible good.
Each agent $i$ has a private valuation $V_i \in \Rplusz$, drawn independently from a distribution $D_i$ on $\Rplusz$.

We define the allocation rule as a measurable map $w:(\Rplusz)^n\to \Delta(\N\cup\{0\})$.
For convenience, we allow a null outcome $0$. Since values are nonnegative, reallocating the mass assigned to $0$ according to a lottery over $\N$ weakly preserves all guaranteed-utility lower bounds.
For a value profile $v\in(\Rplusz)^n$, we write $w_i(v)$ for the probability that agent $i$ receives the good.
In particular, $w_i(v)\ge 0$ for all $i\in\N$ and $\sum\limits_{i\in\N} w_i(v)\le 1$.

The expected utility of agent $i$ depends on the mechanism $w$ and the environment $(\alpha_{\N}, D_{\N})$. We denote it formally as:
\begin{equation}\label{eq:Ui_def}
    U_i(\alpha_{\N}, D_{\N}, w)
    = \E_{V_\N \sim D_{\N}} \big[ V_i \cdot w_i(V_\N) \big].
\end{equation}

\paragraph{The fair-floor target.}
Let $\alpha \in (0,1]$. For a distribution $D$, use the same randomized CDF rank construction as in \eqref{eq:randomized-cdf-rank}: if $V\sim D$ and $Z\sim {\rm Unif}[0,1]$ is independent, then
\[
    R_D(V,Z)=F_D(V-)+Z\bigl(F_D(V)-F_D(V-)\bigr).
\]
This is uniform on the CDF jump interval conditional on $V$, and marginally uniform on $[0,1]$.  We define the fair-floor target level as
\begin{equation}\label{eq:fstar_def}
    f^*(\alpha, D) = \E\left[ V \cdot R_D(V,Z)^{\frac{1}{\alpha} - 1} \right].
\end{equation}

\begin{Lemma}[Feasibility of the fair-floor target]\label{lem:benchmark}
Assume $\alpha_j>0$ for every $j\in\N$, $\sum\limits_{j\in\N}\alpha_j=1$, and that $V_1,\dots,V_n$ are drawn independently with $V_i\sim D_i$.
Consider the stochastic allocation rule that draws $Z_j\sim {\rm Unif}[0,1]$ independently of $V_\N$ and independently across $j$, sets $r_j:=R_{D_j}(V_j,Z_j)$, and assigns the good to
\[
k(V_\N,Z_\N)\in \argmax_{j\in\N} \ r_j^{1/\alpha_j},
\]
with an arbitrary tie-breaking rule (ties occur with probability $0$).
Equivalently, after integrating out $Z_\N$, this stochastic rule defines a measurable allocation kernel $w(v)$ in the sense of \eqref{eq:Ui_def}.
Then for every $i\in\N$ we have
\[
U_i(\alpha_\N,D_\N,w)=f^*(\alpha_i,D_i),
\qquad
\E[w_i(V_\N)]=\alpha_i.
\]
\end{Lemma}

\begin{proof}
Let $r_i:=R_{D_i}(V_i,Z_i)$. By construction, $r_i\sim \mathrm{Unif}[0,1]$, and independence of the pairs $(V_i,Z_i)$ implies independence of the $r_i$'s.
Moreover, for each $j$ the random variable $r_j^{1/\alpha_j}$ has a continuous distribution on $[0,1]$, hence ties occur with probability $0$.

Fix $i\in\N$ and condition on $r_i=s\in[0,1]$. Player $i$ wins iff for all $j\neq i$ we have
$r_j^{1/\alpha_j}\le s^{1/\alpha_i}$, equivalently $r_j\le s^{\alpha_j/\alpha_i}$. Therefore,
\[
\P(i\text{ wins}\mid r_i=s)
=\prod_{j\neq i}\P(r_j\le s^{\alpha_j/\alpha_i})
=\prod_{j\neq i} s^{\alpha_j/\alpha_i}
= s^{\frac{1-\alpha_i}{\alpha_i}}
= s^{\frac{1}{\alpha_i}-1},
\]
where we used $\sum\limits_{j\in\N}\alpha_j=1$.

Finally,
\[
U_i(\alpha_\N,D_\N,w)
=\E\!\big[V_i\,\P(i\text{ wins}\mid r_i)\big]
=\E\!\big[V_i\, r_i^{\frac{1}{\alpha_i}-1}\big]
=f^*(\alpha_i,D_i).
\]
The same conditional probability computation, without the factor $V_i$, gives
\[
    \E[w_i(V_\N)]
    =\int_0^1 s^{1/\alpha_i-1}\,ds
    =\alpha_i,
\]
where $w$ is the kernel obtained after integrating out $Z_\N$.
\end{proof}

For the rest of this section, write
\[
    g_i := f^*(\alpha_i,D_i)
\]
for the fair-floor target level of player \(i\). Lemma~\ref{lem:benchmark} says that \(g=(g_i)_i\) is feasible in the auxiliary allocation problem; for the quota-constrained comparison, it satisfies the prescribed shares exactly.

\subsection{Asymptotic upper bounds}

\paragraph{Optimality convention.}
Fix an instance \((\alpha_\N,D_\N)\) and an admissible class of allocation rules \(\mathcal A\), either the quota-free class or the expected-share relaxation used for the quota-constrained comparison. Let
\[
    K_{\mathcal A}(\alpha_\N,D_\N)
    :=
    \bigl\{(U_i(\alpha_\N,D_\N,w))_{i\in\N}: w\in\mathcal A\bigr\}
\]
be the feasible expected-utility region in the auxiliary allocation problem.
For \(\lambda\ge1\), we say that the target levels \(g_i=f^*(\alpha_i,D_i)\) are \(1/\lambda\)-optimal in the class \(\mathcal A\) if there is no \(u\in K_{\mathcal A}(\alpha_\N,D_\N)\) such that
\[
    u_i>\lambda g_i\qquad\text{for every }i\in\N.
\]
Equivalently, no admissible allocation rule can improve all these target levels by a common factor strictly larger than \(\lambda\).

This is the natural multiplicative comparison for NTU guarantees.
Utility units are player-specific: multiplying player \(i\)'s utility by a positive constant is only a change of representation.
Dually, in supported cases, the infeasibility of a strict common improvement is equivalent to the existence of a positive exchange-rate vector \(\beta\) such that no admissible rule attains more than \(\lambda\sum_i\beta_i g_i\) in \(\beta\)-weighted expected utility.

\begin{Theorem}[Quota-free $\lambdaBar$ upper bound]\label{thm:1.283}
There do not exist weights $\alpha_{\N} > 0$ with $\sum\limits_{i\in\N} \alpha_i = 1$, distributions $D_{\N}$ with finite expectations, and an allocation rule $w:(\Rplusz)^n\to\Delta(\N\cup\{0\})$ such that for every agent $i \in \N$,
\begin{equation}\label{eq:main_inequality_1283}
U_i(\alpha_{\N}, D_{\N}, w) > \lambdaBar \cdot f^*(\alpha_i, D_i).
\end{equation}
\end{Theorem}

By Lemma~\ref{lem:benchmark}, the target vector \(g_i=f^*(\alpha_i,D_i)\) is feasible in the auxiliary allocation problem.
Theorem~\ref{thm:1.283} says that no quota-free feasible expected-utility vector strictly dominates \(\lambdaBar\,g\) coordinatewise.
Hence these fair-floor target levels are \(1/\lambdaBar\)-optimal in the quota-free model, in the sense of the convention above.
The proof is given in Appendix~\ref{app:proofs-constant-bounds}.

\begin{Theorem}[Quota-constrained $\lambdaShare$ upper bound]\label{thm:1.147}
In the expected-share relaxation of the quota-constrained problem, there do not exist weights $\alpha_{\N} > 0$ with $\sum\limits_{i\in\N} \alpha_i = 1$, distributions $D_{\N}$ with finite expectations, and an allocation rule $w:(\Rplusz)^n\to\Delta(\N\cup\{0\})$ satisfying
\[
    \E[w_i(V_\N)]\le\alpha_i\qquad(i\in\N)
\]
such that for every agent $i\in\N$,
\begin{equation}\label{eq:main_inequality_1147}
U_i(\alpha_{\N},D_{\N},w)>\lambdaShare\cdot f^*(\alpha_i,D_i).
\end{equation}
\end{Theorem}

By Lemma~\ref{lem:benchmark}, the target vector \(g_i=f^*(\alpha_i,D_i)\) is feasible and satisfies the prescribed shares exactly.
Theorem~\ref{thm:1.147} says that no feasible expected-utility vector in the expected-share relaxation of the quota-constrained problem strictly dominates \(\lambdaShare\,g\) coordinatewise.
Hence these fair-floor target levels are \(1/\lambdaShare\)-optimal for the quota-constrained comparison.
The proof is completed in Subsection~\ref{app:proofs-expected-share}.

\paragraph{Capacity-envelope guide to the two constants.}
After the rank-space Bernoulli reduction, both upper-bound proofs use the same subset-capacity obstacle. In both cases the obstacle rules out a feasible expected-utility vector that strictly dominates a common rescaling of the fair-floor target vector \(g\). The two constants differ only in the terminal test. Writing
\[
    \Phi_\lambda(t):=\frac{\lambda(1-e^{-t})}{t},
\]
the relevant capacity curve is the obstacle-following trajectory below $h(A)=1-e^{-A}$. In the $y$-parameter it is
\[
    x_\lambda(y)=
    \int_{y_0(\lambda)}^y
    \left(\frac1{t^2}-\frac1{t(e^t-1)}\right)\,\dd t,
    \qquad
    y_0=\lambda(1-e^{-y_0}).
\]
For the expected-share completion it is more convenient to use the residual-capacity parameter $c$, defined by $S=e^{-c}$. If $e^{-c}=\Phi_\lambda(y)$, set
\[
    X_\lambda(c):=x_\lambda(y),
    \qquad
    S_\lambda(X_\lambda(c))=e^{-c}.
\]
In the quota-free completion, feasibility requires this curve to accumulate total normalized mass $x=1$; the critical value is the one for which the trajectory just reaches $x=1$, numerically $\lambda_{\rm crit}\approx\lambdaCritNum$.

For the expected-share relaxation of the quota-constrained variant, the share ceiling gives a maximal tail parameter
\[
    \bar\tau_\lambda=-\log(1-1/\lambda).
\]
If the remaining normalized mass is $m=1-x$, even the most favorable ceiling-tight tail can provide expected utility at most
\[
    S_\lambda(x)\bigl(1-e^{-m\bar\tau_\lambda}\bigr)
    =S_\lambda(x)\bigl(1-(1-1/\lambda)^m\bigr).
\]
Thus tail closure requires
\[
    S_\lambda(x)\ge \widehat S_\lambda(x)
    :=\frac{1-x}{1-(1-1/\lambda)^{1-x}},
\]
with the continuous endpoint value $\widehat S_\lambda(1)=1/\bar\tau_\lambda$.  The quota-constrained critical value is where the capacity curve just touches this threshold, numerically near $\lambdaShareCritNum$; Theorem~\ref{thm:1.147} uses the slightly weaker rounded bound $\lambdaShare$.
Figure~\ref{fig:critical-curves} illustrates the common capacity curve and the two terminal tests.  The rigorous proofs are given in Appendix~\ref{app:proofs-constant-bounds}, with the expected-share completion for the quota-constrained bound in Subsection~\ref{app:proofs-expected-share}.

\begin{figure}[H]
    \centering
    \includegraphics[width=0.86\textwidth]{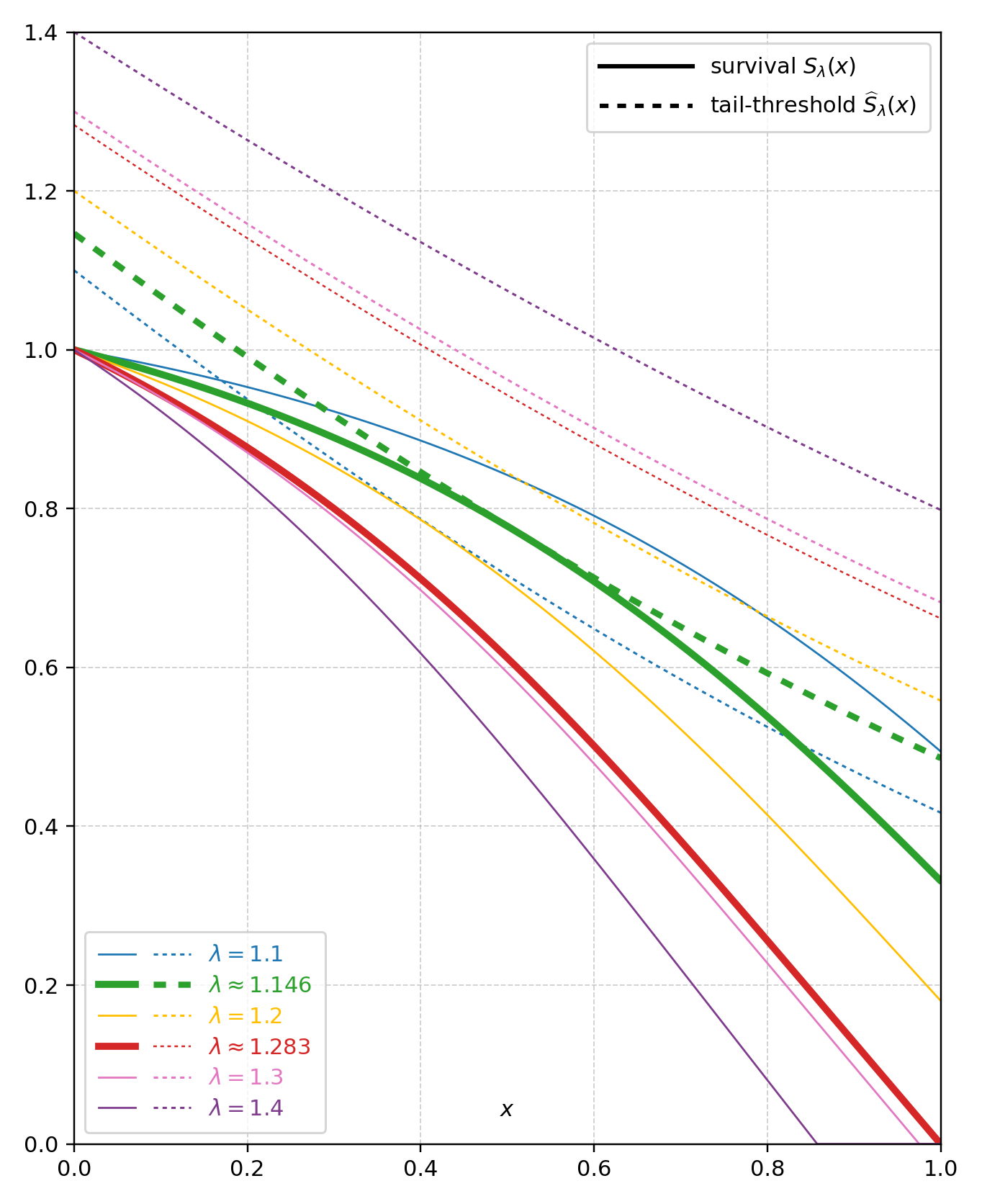}
    \caption{Illustration of the shared capacity-envelope trajectory and the two terminal tests.  Solid curves show $S_\lambda(x)$; dashed curves show the tail-threshold $\widehat S_\lambda(x)$ for the expected-share relaxation of the quota-constrained closure.  The quota-free critical value is where the capacity curve accumulates mass $x=1$ as survival vanishes (red, $\lambda\approx\lambdaCritNum$).  The quota-constrained touching regime is shown in green near $\lambda\approx\lambdaShareCritNum$; Theorem~\ref{thm:1.147} uses the slightly weaker rounded bound $\lambdaShare$. Both highlighted values correspond to infeasibility thresholds for a larger common rescaling of the fair-floor target vector \(g_i=f^*(\alpha_i,D_i)\), in two different feasible classes.}
    \label{fig:critical-curves}
\end{figure}

\section{Open questions}\label{sec:open-questions}

The errors for the transfer-free case are far from optimal.
The following question provides a simple test case for improvement.

\begin{Question} \label{que:gue-uniform}
Consider the repeated good allocation game where the valuations on the good are chosen uniformly from $[0,1]$, independently across periods and players.

A natural candidate class of mechanisms has the following structure.
The mechanism chooses a compensation vector of the players dependent on the report history, and in each round, the player with the highest reported valuation plus compensation gets the good.
Find the best rule for the compensation vector.
\end{Question}

\begin{Question}
    Extend the results for more general utility functions, e.g., allowing risk aversion. To do so, apply the general tendering mechanism in \cite{CsE}.
\end{Question}

\begin{Question}
    Extend the results for costly transfers, or other intermediate cases between TU and NTU.
\end{Question}

\begin{Question}
    For the setup in Theorem~\ref{thm:1.283}, can the freedom in the surplus-splitting rule $c$ improve equilibrium-welfare or price-of-anarchy bounds for scalar objectives other than the common rescaling of the fair-floor target levels?
\end{Question}

\section{Acknowledgement}

I thank Siddhartha Banerjee, András Pongrácz, and Éva Tardos for helpful discussions.

\bibliographystyle{chicago}
\bibliography{PF_GUM}

@misc{CsE,
  author       = {Cs{\'o}ka, Endre},
  title        = {Efficient Teamwork},
  year         = {2006},
  howpublished = {Working paper},
  note         = {arXiv:cs.GT/0602009; MSc thesis, ELTE, Budapest, 2008; presented at the International Conference on Game Theory, Stony Brook, 2015, the 9th Annual Conference of the Israeli Chapter of the Game Theory Society, Haifa, 2017, and the Conference on Economic Design, York, 2017}
}

@unpublished{CLRT,
  author = {Cs{\'o}ka, Endre and Liu, Heng and Rodivilov, Alexander and Teytelboym, Alexander},
  title  = {A Collusion-Proof Efficient Dynamic Mechanism},
  year   = {2024},
  note   = {Working paper, SSRN 4419623},
  url    = {https://ssrn.com/abstract=4419623},
  doi    = {10.2139/ssrn.4419623}
}

@unpublished{GUE25,
  author = {Cs{\'o}ka, Endre and Pongr{\'a}cz, Andr{\'a}s and Rodivilov, Alexander},
  title  = {Guaranteed Utility Equilibrium},
  year   = {2024},
  note   = {Working paper, SSRN 5068726},
  url    = {https://ssrn.com/abstract=5068726},
  doi    = {10.2139/ssrn.5068726}
}

@inproceedings{FBT24,
  author    = {Fikioris, Giannis and Banerjee, Siddhartha and Tardos, {\'E}va},
  title     = {Beyond Worst-Case Online Allocation via Dynamic Max-Min Fairness},
  booktitle = {Proceedings of the ACM Conference on Economics and Computation},
  series    = {EC '25},
  year      = {2025},
  publisher = {Association for Computing Machinery},
  doi       = {10.1145/3736252.3742501},
  url       = {https://doi.org/10.1145/3736252.3742501},
  note      = {Also available as arXiv:2310.08881}
}

@inproceedings{BFT23,
  author    = {Banerjee, Siddhartha and Fikioris, Giannis and Tardos, {\'E}va},
  title     = {Robust Pseudo-Markets for Reusable Public Resources},
  booktitle = {Proceedings of the 24th ACM Conference on Economics and Computation},
  series    = {EC '23},
  pages     = {241},
  numpages  = {1},
  year      = {2023},
  publisher = {Association for Computing Machinery},
  location  = {London, United Kingdom},
  doi       = {10.1145/3580507.3597723},
  url       = {https://doi.org/10.1145/3580507.3597723},
  note      = {Also available as arXiv:2302.09127}
}

@inproceedings{GBI21,
  author    = {Gorokh, Artur and Banerjee, Siddhartha and Iyer, Krishnamurthy},
  title     = {The Remarkable Robustness of the Repeated {F}isher Market},
  booktitle = {Proceedings of the 22nd ACM Conference on Economics and Computation},
  series    = {EC '21},
  pages     = {562},
  numpages  = {1},
  year      = {2021},
  publisher = {Association for Computing Machinery},
  location  = {Budapest, Hungary},
  note      = {Working paper, SSRN 3411444}
}

@inproceedings{GBI17,
  author    = {Gorokh, Artur and Banerjee, Siddhartha and Iyer, Krishnamurthy},
  title     = {From Monetary to Non-Monetary Mechanism Design via Artificial Currencies},
  booktitle = {Proceedings of the 2017 ACM Conference on Economics and Computation},
  series    = {EC '17},
  pages     = {563--564},
  numpages  = {2},
  year      = {2017},
  publisher = {Association for Computing Machinery},
  location  = {Cambridge, Massachusetts, USA},
  note      = {Working paper, SSRN 2964082}
}

\newpage
\appendix

\section{Capacity proofs for the single-good upper bounds}\label{app:proofs-constant-bounds}

This appendix proves Theorems~\ref{thm:1.283} and~\ref{thm:1.147}.  The two arguments share the same rank-space Bernoulli reduction and the same subset-capacity obstacle.  They diverge only at the terminal condition: the quota-free proof asks whether the capacity-envelope trajectory can be continued to total player mass one, while the quota-constrained proof, via the expected-share relaxation, has the additional ceiling imposed by the share constraints.

\subsection{Rank-space and Bernoulli reduction}\label{app:common-rank-bernoulli}

Fix weights $\alpha_i>0$ with $\sum_i\alpha_i=1$, independent nonnegative values $V_i\sim D_i$ with finite expectations, and an allocation rule
\[
    w:(\Rplusz)^n\to\Delta(\N\cup\{0\}).
\]
For a distribution $D$ write $Q_D$ for its quantile function and use the randomized-rank transform $R_D(V,Z)$ from \eqref{eq:randomized-cdf-rank}.  Passing to ranks gives independent uniform variables $R_i$ and an induced rank rule $\widehat w:[0,1]^n\to\Delta(\N\cup\{0\})$.  Define the interim curve
\[
    x_i(s):=\int_{[0,1]^{n-1}}\widehat w_i(s,r_{-i})\,\dd r_{-i}.
\]
Then
\begin{equation}\label{eq:cap-rank-representation}
    U_i=\int_0^1 Q_{D_i}(s)x_i(s)\,\dd s,
    \qquad
    f^*(\alpha_i,D_i)=\int_0^1 Q_{D_i}(s)s^{1/\alpha_i-1}\,\dd s.
\end{equation}
If expected-share constraints are present, they become
\begin{equation}\label{eq:cap-share-rank}
    \int_0^1 x_i(s)\,\dd s\le \alpha_i.
\end{equation}
Thus changing only the quantile functions while keeping the rank rule fixed preserves all expected-share constraints.
After the Bernoulli replacement, the rank-space rule is allowed to depend on the randomized rank inside atoms. This only enlarges the admissible class of allocation rules, and hence is harmless for an upper-bound/impossibility proof. From this point on, we suppress hats and write $w$ also for the rank-space rule; expectations in the Bernoulli reductions are over the independent rank variables.

\begin{Proposition}[Bernoulli reduction]\label{prop:cap-bernoulli-reduction}
Fix $\lambda>1$.  Suppose that
\begin{equation}\label{eq:cap-counterexample-general}
    U_i>\lambda f^*(\alpha_i,D_i)\qquad(i\in\N).
\end{equation}
If expected-share constraints \eqref{eq:cap-share-rank} are also imposed, assume them as well.  Then there is another rank-space instance, with the same weights and the same rank rule, in which every value is Bernoulli, $V_i\in\{0,1\}$ with $0<\P(V_i=1)<1$, and \eqref{eq:cap-counterexample-general} still holds.  In the expected-share case, the constraints \eqref{eq:cap-share-rank} are preserved.
\end{Proposition}

\begin{proof}
Fix player $i$, keeping the rank rule and all other quantiles fixed.  Put
\[
    H_i(s):=x_i(s)-\lambda s^{1/\alpha_i-1}.
\]
The strict inequality for player $i$ says
\[
    \int_0^1 Q_{D_i}(s)H_i(s)\,\dd s>0.
\]
By the layer-cake representation for nonnegative quantiles,
\[
    Q_{D_i}(s)=\int_0^\infty \one\{Q_{D_i}(s)>t\}\,\dd t.
\]
For each $t$, the set $\{s:Q_{D_i}(s)>t\}$ is an upper interval, up to null sets.  Hence $\one\{Q_{D_i}(s)>t\}$ is a Bernoulli threshold quantile $Q^{(p)}(s)=\one\{s>1-p\}$ for some $p\in[0,1]$.  Since the integral against $H_i$ is positive, some threshold gives a positive value.  The endpoint $p=0$ gives zero.  If only $p=1$ gave a positive value, continuity of
\[
    p\mapsto \int_{1-p}^1 H_i(s)\,\dd s
\]
would give a nearby $p<1$ with positive value.  Thus one can choose $p\in(0,1)$.

Replacing $D_i$ by this Bernoulli law preserves the rank rule and the interim curves of all players.  Therefore the other players' inequalities are unchanged, player $i$'s inequality is preserved, and any expected-share constraints remain valid.  Repeating this for all players gives the claim.
\end{proof}

For a Bernoulli instance we write
\begin{equation}\label{eq:cap-p-tau}
    p_i:=\P(V_i=1)=1-e^{-\alpha_i\tau_i},
    \qquad 0<\tau_i<\infty.
\end{equation}
Then
\begin{equation}\label{eq:cap-bernoulli-benchmark}
    f^*(\alpha_i,D_i)=\alpha_i(1-e^{-\tau_i}).
\end{equation}

\begin{Lemma}[Subset capacity]\label{lem:cap-subset-capacity}
In a Bernoulli rank-space instance satisfying
\[
    U_i\ge \lambda\alpha_i(1-e^{-\tau_i})\qquad(i\in\N),
\]
one has, for every $S\subseteq\N$,
\begin{equation}\label{eq:cap-subset-capacity}
    \lambda\sum_{i\in S}\alpha_i(1-e^{-\tau_i})
    \le
    1-\exp\left(-\sum_{i\in S}\alpha_i\tau_i\right).
\end{equation}
\end{Lemma}

\begin{proof}
Let $R_\N$ be the independent uniform rank vector and write
\[
    V_i(R_i)=\one\{R_i>1-p_i\}.
\]
For every rank realization $r$,
\[
    \sum_{i\in S}V_i(r_i)w_i(r)
    \le
    \one\{\exists i\in S: V_i(r_i)=1\},
\]
because at most one unit of allocation probability can be assigned to agents in $S$.  Taking expectations gives
\[
    \sum_{i\in S}U_i
    \le
    \P(\exists i\in S:V_i=1)
    =1-\prod_{i\in S}(1-p_i).
\]
Using \eqref{eq:cap-p-tau}, the right-hand side is
\[
    1-\exp\left(-\sum_{i\in S}\alpha_i\tau_i\right).
\]
Combining this with the lower bounds on $U_i$ proves \eqref{eq:cap-subset-capacity}.
\end{proof}

\subsection{The common capacity envelope}\label{app:capacity-envelope}

For $\lambda>1$ define
\begin{equation}\label{eq:cap-Phi-def}
    \Phi_\lambda(t):=\frac{\lambda(1-e^{-t})}{t}\qquad(t>0).
\end{equation}
Let $y_0(\lambda)$ be the positive solution of
\begin{equation}\label{eq:cap-y0-equation}
    y_0=\lambda(1-e^{-y_0}),
\end{equation}
which is equivalent to $\Phi_\lambda(y_0)=1$.

The subset-capacity inequality can be viewed as saying that sorted capacity atoms must remain below the obstacle
\[
    h(A):=1-e^{-A}.
\]
The extremal continuous envelope follows the obstacle.  In the parameter $y$ determined by $e^{-A}=\Phi_\lambda(y)$, the accumulated normalized player mass is
\begin{equation}\label{eq:cap-x-y-def}
    x_\lambda(y):=
    \int_{y_0(\lambda)}^y
    \left(\frac1{t^2}-\frac1{t(e^t-1)}\right)\,\dd t.
\end{equation}
Equivalently, writing $S=e^{-A}$ and $x=x_\lambda(y)$ gives the same curve that is displayed as $S_\lambda(x)$ in Figure~\ref{fig:critical-curves}.

\subsection{Quota-free completion and the \texorpdfstring{$\lambdaBar$}{1.283} bound}\label{app:unconstrained-completion}

The quota-free completion uses only the prefix consequences of the subset-capacity inequality.

\begin{Lemma}[Prefix majorization]\label{lem:cap-prefix-majorization}
Fix $\lambda>1$.  Let $\alpha_i>0$, $\sum_i\alpha_i=1$, and
\[
    0<\tau_1\le\tau_2\le\cdots\le\tau_n<\infty.
\]
Assume that for every prefix $\{1,\ldots,k\}$,
\begin{equation}\label{eq:cap-prefix-condition}
    \lambda\sum_{i\le k}\alpha_i(1-e^{-\tau_i})
    \le
    1-\exp\left(-\sum_{i\le k}\alpha_i\tau_i\right).
\end{equation}
Define
\begin{equation}\label{eq:cap-xcrit-def}
    x_{\rm crit}(\lambda):=
    \int_{y_0(\lambda)}^\infty
    \left(\frac1{y^2}-\frac1{y(e^y-1)}\right)\,\dd y.
\end{equation}
Then
\[
    1\le x_{\rm crit}(\lambda).
\]
\end{Lemma}

\begin{proof}
Set
\[
    A_k:=\sum_{i\le k}\alpha_i\tau_i,
    \qquad A_0:=0,
\]
and define a step function $r:[0,\infty)\to[0,1]$ by
\[
    r(a):=\frac{1-e^{-\tau_i}}{\tau_i}\quad\text{for }a\in[A_{i-1},A_i),
    \qquad
    r(a):=0\quad\text{for }a\ge A_n.
\]
The function $y\mapsto(1-e^{-y})/y$ is strictly decreasing, so $r$ is nonincreasing.  Moreover,
\[
    \int_0^{A_k}r(a)\,\dd a=
    \sum_{i\le k}\alpha_i(1-e^{-\tau_i}).
\]
Thus \eqref{eq:cap-prefix-condition} implies
\begin{equation}\label{eq:cap-prefix-domination}
    \int_0^a r(t)\,\dd t\le \frac{1-e^{-a}}{\lambda}\qquad(a\ge0).
\end{equation}
Indeed this is immediate at the prefix endpoints; it holds between endpoints because the left-hand side is affine and the right-hand side is concave, and it remains true after $A_n$ because the right-hand side is increasing.

Let
\[
    r_\lambda(a):=\frac{e^{-a}}{\lambda}.
\]
Then \eqref{eq:cap-prefix-domination} says that the primitive of $r$ is dominated by that of $r_\lambda$.  Since $r$ is nonincreasing, for every continuity point of $r$ one also has $r(a)\le1/\lambda$.

Define $\psi:[0,1)\to[0,\infty)$ by
\[
    \psi\left(\frac{1-e^{-y}}{y}\right)=\frac1y\qquad(y>0),
\]
and set $\psi(0)=0$.  The function $\psi$ is increasing and convex.  Indeed, with $z=1/y$,
\[
    \frac{1-e^{-y}}{y}=z(1-e^{-1/z})=:h(z),
\]
where
\[
    h'(z)=1-e^{-1/z}\left(1+\frac1z\right)>0,
    \qquad
    h''(z)=-\frac{e^{-1/z}}{z^3}<0.
\]
Thus $h$ is increasing and concave, and its inverse $\psi$ is increasing and convex.

Since $\psi(r(a))=1/\tau_i$ on $[A_{i-1},A_i)$,
\begin{equation}\label{eq:cap-prefix-mass}
    \int_0^\infty \psi(r(a))\,\dd a
    =\sum_i\alpha_i=1.
\end{equation}
For nonincreasing functions, primitive domination implies increasing-convex domination.  Concretely, for each $t\ge0$,
\[
    \int_0^\infty (r(a)-t)_+\,\dd a
    =\sup_{b\ge0}\left(\int_0^b r(a)\,\dd a-bt\right),
\]
and the same identity holds for $r_\lambda$.  Hence \eqref{eq:cap-prefix-domination} implies
\[
    \int_0^\infty (r(a)-t)_+\,\dd a
    \le
    \int_0^\infty (r_\lambda(a)-t)_+\,\dd a.
\]
Representing the increasing convex function $\psi$ on the compact interval $[0,1/\lambda]$ by an affine term plus positive combinations of the hinge functions $(x-t)_+$ gives
\begin{equation}\label{eq:cap-prefix-convex-domination}
    \int_0^\infty \psi(r(a))\,\dd a
    \le
    \int_0^\infty \psi(r_\lambda(a))\,\dd a.
\end{equation}

It remains to compute the right-hand side.  Let $y=y(a)$ be defined by
\[
    \frac{1-e^{-y}}{y}=\frac{e^{-a}}{\lambda}.
\]
At $a=0$ this gives \eqref{eq:cap-y0-equation}.  Also
\[
    a=\log y-\log\lambda-\log(1-e^{-y}),
    \qquad
    \frac{\dd a}{\dd y}=\frac1y-\frac1{e^y-1}.
\]
Since $\psi(r_\lambda(a))=1/y(a)$,
\[
    \int_0^\infty \psi(r_\lambda(a))\,\dd a
    =
    \int_{y_0(\lambda)}^\infty
    \frac1y\left(\frac1y-\frac1{e^y-1}\right)\,\dd y
    =x_{\rm crit}(\lambda).
\]
Combining this with \eqref{eq:cap-prefix-mass} and \eqref{eq:cap-prefix-convex-domination} proves the lemma.
\end{proof}

\begin{Proposition}[Numerical verification at $\lambda=\lambdaBar$]\label{prop:cap-numerical-1283}
For $\lambda=\lambdaBar$,
\[
    x_{\rm crit}(\lambda)<1.
\]
\end{Proposition}

\begin{proof}
Let $\lambda=\lambdaBar$, and let $y_0$ be the positive root of $y=\lambda(1-e^{-y})$.  For $\varphi(y):=\lambda(1-e^{-y})-y$, interval evaluation gives
\[
    \varphi(0.52096)>0>\varphi(0.52097),
\]
and $\varphi$ is strictly decreasing on $[0.5,\infty)$.  Hence $y_0\in[a,b]:=[0.52096,0.52097]$.
Using
\[
    \frac1{e^y-1}=\sum_{m=1}^\infty e^{-my},
    \qquad
    E_1(z):=\int_z^\infty \frac{e^{-t}}{t}\,\dd t,
\]
monotonicity gives
\[
    x_{\rm crit}(\lambda)
    =\frac1{y_0}-\sum_{m=1}^\infty E_1(my_0)
    \le
    \frac1a-\sum_{m=1}^{13}E_1(mb).
\]
A direct interval-arithmetic evaluation of these thirteen one-dimensional integrals gives
\[
    \sum_{m=1}^{13}E_1(mb)>0.91959646,
    \qquad
    \frac1a<1.91953318.
\]
Therefore
\[
    x_{\rm crit}(\lambda)<1.91953318-0.91959646<0.99993672<1.
\]
\end{proof}

\begin{proof}[Proof of Theorem~\ref{thm:1.283}]
Assume a strict counterexample exists.  If $f^*(\alpha_i,D_i)=0$ for some $i$, then $V_i=0$ almost surely and $U_i=0$, contradicting strictness.  Hence choose
\[
    \lambdaBar<\lambda<\min_i\frac{U_i(\alpha_\N,D_\N,w)}{f^*(\alpha_i,D_i)}.
\]
By Proposition~\ref{prop:cap-bernoulli-reduction}, reduce to a Bernoulli counterexample at this same $\lambda$.  Relabel players so that $\tau_1\le\cdots\le\tau_n$.  Applying Lemma~\ref{lem:cap-subset-capacity} to every prefix gives the hypotheses of Lemma~\ref{lem:cap-prefix-majorization}, hence
\[
    1\le x_{\rm crit}(\lambda).
\]
On the other hand, $x_{\rm crit}$ is strictly decreasing in $\lambda$: the root $y_0(\lambda)$ is strictly increasing in $\lambda$, and the integrand in \eqref{eq:cap-xcrit-def} is strictly positive.  Therefore
\[
    x_{\rm crit}(\lambda)<x_{\rm crit}(\lambdaBar)<1
\]
by Proposition~\ref{prop:cap-numerical-1283}, a contradiction.
\end{proof}

\subsection{Expected-share completion for the quota-constrained \texorpdfstring{$\lambdaShare$}{1.147} bound}\label{app:proofs-expected-share}

This subsection proves Theorem~\ref{thm:1.147}.  The argument uses only the necessary subset-capacity inequalities and a convex-majorization comparison with the obstacle $h(A)=1-e^{-A}$; no normal form for allocation rules is used.

Throughout this subsection set
\[
    \lambda:=\lambdaShare.
\]
Call an expected-share instance a $\lambda$-counterexample if it satisfies
\[
    U_i(\alpha_\N,D_\N,w)>\lambda f^*(\alpha_i,D_i)\qquad(i\in\N)
\]
and $\E[w_i(V_\N)]\le\alpha_i$ for every $i$.

Let
\[
    \bar\tau=\bar\tau_\lambda:=-\log(1-1/\lambda).
\]
For $\lambda=\lambdaShare$ one has $\bar\tau>1$, since $\lambdaShare<e/(e-1)$. Moreover
\[
    \Phi_\lambda(\bar\tau)=\frac1{\bar\tau}<1=\Phi_\lambda(y_0(\lambda)).
\]
Since $\Phi_\lambda$ is strictly decreasing, this implies $0<y_0(\lambda)<\bar\tau$.

\begin{Lemma}[Capacity-atom reduction]\label{lem:es-capacity-reduction}
If a $\lambda$-counterexample exists, then there are $L>0$ and a nondecreasing step function
\[
    \tau:[0,L]\to[y_0(\lambda),\bar\tau]
\]
such that, with
\[
    B(A):=\int_0^A\Phi_\lambda(\tau(s))\,\dd s,
    \qquad
    h(A):=1-e^{-A},
\]
one has
\begin{equation}\label{eq:es-obstacle}
    B(A)\le h(A)\qquad(0\le A\le L),
\end{equation}
and
\begin{equation}\label{eq:es-mass-one}
    \int_0^L\frac{\dd s}{\tau(s)}=1.
\end{equation}
\end{Lemma}

\begin{proof}
By Proposition~\ref{prop:cap-bernoulli-reduction}, a counterexample may be assumed to have Bernoulli values. Write $p_i=1-e^{-\alpha_i\tau_i}$ as in \eqref{eq:cap-p-tau}. We are using the rank-space rule notation fixed above. Since withholding is allowed, move all allocation mass assigned to inactive players to outcome $0$. Utilities are unchanged and expected shares weakly decrease. Hence, after this no-waste modification,
\[
    U_i=\E[V_iw_i]=\E[w_i]\le\alpha_i,
\]
because $w_i$ is positive only on profiles where player $i$ is active.
A counterexample therefore satisfies
\[
    U_i>\lambda\alpha_i(1-e^{-\tau_i}),
    \qquad
    U_i\le\alpha_i,
\]
and consequently $\tau_i<\bar\tau$ for every $i$.

Set
\[
    a_i:=\alpha_i\tau_i.
\]
Then
\[
    \lambda\alpha_i(1-e^{-\tau_i})=a_i\Phi_\lambda(\tau_i),
    \qquad
    \sum_i\frac{a_i}{\tau_i}=\sum_i\alpha_i=1.
\]
For every nonempty $T\subseteq\N$, Lemma~\ref{lem:cap-subset-capacity} and the strict counterexample inequalities give
\begin{equation}\label{eq:es-subset-capacity}
    \sum_{i\in T}a_i\Phi_\lambda(\tau_i)
    <
    1-\exp\left(-\sum_{i\in T}a_i\right).
\end{equation}
Only this necessary capacity condition is used in the upper-bound proof.  Applying it to $T=\{i\}$ gives
\[
    \Phi_\lambda(\tau_i)<\frac{1-e^{-a_i}}{a_i}<1,
\]
so $\tau_i>y_0(\lambda)$.  Thus the atoms lie in the open interval $(y_0(\lambda),\bar\tau)$; below we use the closed interval $[y_0(\lambda),\bar\tau]$ only to simplify notation.

Choose a random subset $T$ by including $i$ independently with probability $\rho_i\in[0,1]$.  For $T=\varnothing$ add the trivial equality $0=0$ to \eqref{eq:es-subset-capacity}.  Averaging and using the concavity of $h(A)=1-e^{-A}$ gives
\begin{equation}\label{eq:es-fractional-capacity}
    \sum_i\rho_i a_i\Phi_\lambda(\tau_i)
    \le
    1-\exp\left(-\sum_i\rho_i a_i\right)
    \qquad(0\le\rho_i\le1).
\end{equation}

Order the atoms by nondecreasing $\tau_i$.  Since $\Phi_\lambda$ is decreasing, this is the nonincreasing order of the densities $\Phi_\lambda(\tau_i)$.  Let $\tau(s)$ be the step function which takes value $\tau_i$ on an interval of length $a_i$, in this order, and let $L:=\sum_i a_i$.  For each fixed width $A$, the left-hand side of \eqref{eq:es-fractional-capacity} is maximized by taking the first $A$ units of this sorted list, with a fractional last atom if needed.  This gives \eqref{eq:es-obstacle}.  The identity \eqref{eq:es-mass-one} is just $\sum_i a_i/\tau_i=1$.
\end{proof}

Let
\begin{equation}\label{eq:es-p0-cmax}
    p_0:=\Phi_\lambda(\bar\tau)=\frac1{\bar\tau},
    \qquad
    c_{\max}:=\log\bar\tau>0.
\end{equation}
For $p\in[p_0,1]$ define
\[
    \Psi_\lambda(p):=\frac1{\Phi_\lambda^{-1}(p)}.
\]
Finally set
\begin{equation}\label{eq:es-X-def}
    X_\lambda(c):=\int_0^c\Psi_\lambda(e^{-A})\,\dd A,
    \qquad 0\le c\le c_{\max}.
\end{equation}
The elementary calculus facts used below are recorded at the end of the subsection; in particular, $\Psi_\lambda$ is convex on $[p_0,1]$.

\begin{Lemma}[Continuous convex majorization]\label{lem:es-convex-majorization}
Let $r,r^*\in L^1([0,L])$ be nonnegative nonincreasing functions, with values in a compact interval $I\subset[0,\infty)$, with equal integrals and
\[
    \int_0^t r(A)\,\dd A\le \int_0^t r^*(A)\,\dd A\qquad(0\le t\le L).
\]
Then for every convex function $\psi$ on $I$,
\[
    \int_0^L\psi(r(A))\,\dd A
    \le
    \int_0^L\psi(r^*(A))\,\dd A.
\]
\end{Lemma}

\begin{proof}
It is enough to prove the claim for hinge functions $(x-t)_+$ and affine functions, since any convex function on the compact interval $I$ has the standard affine-plus-positive-hinges representation.  The affine part is unchanged because the total integrals are equal.  For a hinge,
\[
    \int_0^L(r(A)-t)_+\,\dd A
    =\sup_{0\le b\le L}\left(\int_0^b r(A)\,\dd A-bt\right),
\]
because $r$ is nonincreasing; the same identity holds for $r^*$.  The assumed prefix domination gives the required hinge inequality.
\end{proof}

\begin{Lemma}[Obstacle lemma]\label{lem:es-obstacle}
Let $\tau:[0,L]\to[y_0(\lambda),\bar\tau]$ be a nondecreasing step function satisfying \eqref{eq:es-obstacle}.  Put
\[
    M:=\int_0^L\frac{\dd s}{\tau(s)}.
\]
Then there are $c\in[0,c_{\max}]$ and $m\ge0$ such that
\begin{equation}\label{eq:es-obstacle-mass}
    M\le X_\lambda(c)+m,
\end{equation}
and
\begin{equation}\label{eq:es-tail-ineq}
    m\le e^{-c}(1-e^{-m\bar\tau}).
\end{equation}
\end{Lemma}

\begin{proof}
Write
\[
    p(A):=\Phi_\lambda(\tau(A)).
\]
Since $\tau$ is nondecreasing and $\Phi_\lambda$ is decreasing, $p$ is nonincreasing.  Moreover $p_0\le p(A)\le1$.  Subtract the minimal slope:
\[
    \widetilde B(A):=B(A)-p_0A,
    \qquad
    \widetilde h(A):=h(A)-p_0A.
\]
Then $\widetilde B\le\widetilde h$ and
\[
    r(A):=\widetilde B'(A)=p(A)-p_0\ge0
\]
is nonincreasing.  Let
\[
    R:=\widetilde B(L)=\int_0^L r(A)\,\dd A.
\]
The function
\[
    \widetilde h(A)=1-e^{-A}-p_0A
\]
is concave and satisfies $\widetilde h'(A)=e^{-A}-p_0$. Hence it increases on $[0,c_{\max}]$, where $c_{\max}=\log\bar\tau=-\log p_0$, and reaches its maximum there. Choose $c$ on this increasing branch such that
\[
    \widetilde h(c)=R.
\]
This is possible because $0\le R=\widetilde B(L)\le\widetilde h(L)$, so $\widetilde h(L)\ge0$, and $\widetilde h(L)\le\max_A\widetilde h(A)=\widetilde h(c_{\max})$. It also gives $c\le L$. Indeed, if $L\le c_{\max}$ this follows from the monotonicity of $\widetilde h$ on $[0,c_{\max}]$, and if $L\ge c_{\max}$ then $c\le c_{\max}\le L$.

Define
\[
    r^*(A):=
    \begin{cases}
        e^{-A}-p_0,&0\le A\le c,\\
        0,&c<A\le L.
    \end{cases}
\]
Then $r^*$ and $r$ have the same integral over $[0,L]$.  Moreover, $r^*$ majorizes $r$: for $t\le c$,
\[
    \int_0^t r(A)\,\dd A
    =\widetilde B(t)
    \le\widetilde h(t)
    =\int_0^t r^*(A)\,\dd A,
\]
and for $t\ge c$,
\[
    \int_0^t r(A)\,\dd A
    \le R
    =\int_0^t r^*(A)\,\dd A.
\]
Applying Lemma~\ref{lem:es-convex-majorization} to the convex function $z\mapsto\Psi_\lambda(p_0+z)$ gives
\[
    M
    =\int_0^L\Psi_\lambda(p_0+r(A))\,\dd A
    \le
    \int_0^L\Psi_\lambda(p_0+r^*(A))\,\dd A.
\]
On $[0,c]$, $p_0+r^*(A)=e^{-A}$; on $(c,L]$, $p_0+r^*(A)=p_0$.  Thus
\[
    M
    \le
    \int_0^c\Psi_\lambda(e^{-A})\,\dd A
    +(L-c)\Psi_\lambda(p_0)
    =X_\lambda(c)+\frac{L-c}{\bar\tau}.
\]
Put
\[
    m:=\frac{L-c}{\bar\tau}.
\]
This proves \eqref{eq:es-obstacle-mass}.

It remains to prove \eqref{eq:es-tail-ineq}.  From the choice of $c$ and $\widetilde B(L)\le\widetilde h(L)$,
\[
    \widetilde h(c)\le\widetilde h(L).
\]
Equivalently,
\[
    p_0(L-c)
    \le h(L)-h(c)
    =e^{-c}(1-e^{-(L-c)}).
\]
Since $p_0=1/\bar\tau$ and $L-c=m\bar\tau$, this is exactly \eqref{eq:es-tail-ineq}.
\end{proof}

The numerical verification below shows $X_\lambda(c_{\max})<1$. Since $X_\lambda$ is increasing, the denominator in the next expression is positive throughout $[0,c_{\max}]$.
For $0\le c\le c_{\max}$ define
\begin{equation}\label{eq:es-Gamma-def}
    \Gamma_\lambda(c):=
    \frac{1-X_\lambda(c)}{1-\exp(-(1-X_\lambda(c))\bar\tau)}
    -e^{-c}.
\end{equation}
The next lemma is the only numerical input for the expected-share bound.

\begin{Lemma}[Numerical verification at $\lambda=\lambdaShare$]\label{lem:es-numerical-verification}
For $\lambda=\lambdaShare$,
\begin{equation}\label{eq:es-Xmax-num}
    X_\lambda(c_{\max})<1,
\end{equation}
and
\begin{equation}\label{eq:es-Gamma-positive}
    \Gamma_\lambda(c)>0\qquad(0\le c\le c_{\max}).
\end{equation}
Numerically,
\[
    \bar\tau\approx2.054472530350634,
    \qquad
    c_{\max}\approx0.720019138772815,
\]
\[
    X_\lambda(c_{\max})\approx0.850921358504375<1,
\]
and an interval enclosure gives
\[
    \min_{0\le c\le c_{\max}}\Gamma_\lambda(c)>0.0018100052.
\]
The minimizing box is centered near
\[
    c\approx0.262356574573656,
    \qquad
    X_\lambda(c)\approx0.511699045199960.
\]
\end{Lemma}

\begin{proof}
The functions involved are explicit one-variable functions.  The stated values are obtained by interval arithmetic: enclose the root of $\Phi_\lambda(y)=1$, evaluate the integral \eqref{eq:es-X-def} with outward rounding, and minimize \eqref{eq:es-Gamma-def} on $[0,c_{\max}]$ by interval subdivision and interval-Newton verification.  Thus the displayed lower bound is verified by interval arithmetic for the explicit one-dimensional function $\Gamma_\lambda$.
\end{proof}

\begin{proof}[Proof of Theorem~\ref{thm:1.147}]
Assume a strict expected-share counterexample exists at $\lambda=\lambdaShare$.  If $f^*(\alpha_i,D_i)=0$ for some $i$, then $V_i=0$ almost surely and $U_i=0$, contradicting strictness.  By Lemma~\ref{lem:es-capacity-reduction}, the counterexample gives a capacity path $\tau$ satisfying \eqref{eq:es-obstacle} and having total mass $M=1$.  Apply Lemma~\ref{lem:es-obstacle}.  There exist $c\in[0,c_{\max}]$ and $m\ge0$ such that
\[
    1\le X_\lambda(c)+m,
    \qquad
    m\le e^{-c}(1-e^{-m\bar\tau}).
\]
Since $X_\lambda$ is increasing, Lemma~\ref{lem:es-numerical-verification} gives $X_\lambda(c)\le X_\lambda(c_{\max})<1$.  Hence
\[
    1-X_\lambda(c)\le m.
\]
For fixed $c$, the function
\[
    t\mapsto e^{-c}(1-e^{-t\bar\tau})-t
\]
is concave, is zero at $t=0$, and is nonnegative at $t=m$.  It is therefore nonnegative at every $0\le t\le m$, in particular at $t=1-X_\lambda(c)$.  Thus
\[
    1-X_\lambda(c)
    \le
    e^{-c}\left(1-e^{-(1-X_\lambda(c))\bar\tau}\right).
\]
This is equivalent to $\Gamma_\lambda(c)\le0$, contradicting Lemma~\ref{lem:es-numerical-verification}.  Therefore no counterexample exists at $\lambda=\lambdaShare$, and Theorem~\ref{thm:1.147} follows.
\end{proof}

\paragraph{Calculus facts for the expected-share completion.}
The function $\Phi_\lambda$ is strictly decreasing on $(0,\infty)$.  Moreover,
\[
    \Psi_\lambda(p)=\frac1{\Phi_\lambda^{-1}(p)}
\]
is convex on $[p_0,1]$.  Indeed, with $p=\Phi_\lambda(y)$,
\[
    \Psi_\lambda''(p)
    =
    -\frac{y^3e^{2y}}
    {\lambda^2(y-e^y+1)^3}>0,
\]
because $y-e^y+1<0$ for $y>0$.

The numerical verification may equivalently be evaluated in the $y$-parameter.  If $e^{-c}=\Phi_\lambda(y)$, then
\[
    X_\lambda(c)
    =
    \int_{y_0(\lambda)}^y
    \left(
        \frac1{t^2}-\frac1{t(e^t-1)}
    \right)\,\dd t.
\]
This is the same capacity curve as in \eqref{eq:cap-x-y-def}, written in the residual-capacity parameter $c$.

\end{document}